\documentclass[a4paper,reqno,oneside]{amsart}
\pdfoutput=1

\usepackage{preamble}
\usepackage[foot]{amsaddr}

\title{Einstein-Yang-Mills wormholes haunted by a phantom field}
\author{Marko Sobak}
\address[Marko Sobak]{Faculty of Mathematics, University of Vienna, Oskar-Morgenstern-Platz 1, 1090 Vienna, Austria}
\email{marko.sobak@univie.ac.at}
\keywords{Wormholes, Einstein-Yang-Mills equations, phantom field}
\begin{document}

% !TEX root = ms.tex

\begin{abstract}
In this article, we study wormhole spacetimes in the framework of the static spherically symmetric $\SU 2$ Einstein-Yang-Mills theory coupled to a phantom scalar field. We show rigorously the existence of an infinite sequence of symmetric wormhole solutions, labelled by the number of zeros of the Yang-Mills potential. These solutions have previously been discovered numerically. Mathematically, the problem resembles the pure Einstein-Yang-Mills system for black hole initial conditions, which was well-studied in the 90s. The main difference in the present work is that the coupling to the phantom field adds a non-trivial degree of complexity to the analysis. After proving the existence of the symmetric wormhole solutions, we also present numerical evidence for the existence of asymmetric ones.
\end{abstract}
    
\maketitle

\section{Introduction}
% !TEX root = ms.tex

\subsection{Wormholes}

Wormholes are hypothetical models of spacetime that describe traversable tunnels connecting asymptotically flat universes. 
Possibly the simplest example of a wormhole is the so-called \textit{Ellis-Bronnikov wormhole}, modelled as the static spherically symmetric spacetime $\rn^2\times\mathbb S^2$ endowed with the Lorentzian metric
\begin{equation}\label{eq-ellis-bronnikov}
g = - \mathrm dt \otimes \mathrm dt + r_0^2 \cosh^2(\rho) (\mathrm d\rho \otimes \mathrm d\rho + \gS),
\end{equation} 
where $\gS$ denotes the standard round metric of the sphere and $r_0>0$ is a free parameter.
The metric $g$ is asymptotically flat with two ends as $\rho \to\pm\infty$, which is the main characterizing condition of wormholes.
Static spherically symmetric wormholes were defined more generally in the pioneering work of Morris and Thorne \cite{morris-thorne}, and many authors have studied them since, see e.g.\ the extensive book of Visser \cite{visser}.
Throughout this paper, we will employ the following definition of wormholes in the static spherically symmetric setting, which is essentially equivalent to the definitions of the aforementioned authors.

\begin{definition}\label{def-wh}
A (static spherically symmetric) \textit{wormhole} is a spacetime with topology $\rn^2 \times \mathbb S^2$, endowed with a Lorentzian metric of the form
\begin{equation}\label{eq-wh-metric}
g = -e^{2\tau(\rho)} \, \mathrm dt \otimes \mathrm dt + r(\rho)^2 \left( \mathrm d\rho \otimes \mathrm d\rho + \gS \right),
\end{equation}
where $\tau,r \in C^\infty(\rn)$ with $r > 0$, such that $g$ is \textit{asymptotically flat with two ends}, so that 
\begin{equation*}
\lim_{\rho\to\pm\infty} r(\rho) = \infty, \quad 
\lim_{\rho\to\pm\infty} m(\rho) = m_\infty^\pm, \quad  
\lim_{\rho\to\pm\infty} \tau(\rho) = \tau_\infty^\pm,
\end{equation*}
where $m_\infty^\pm, \tau_\infty^\pm \in \rn$ are finite, and $m$ denotes the \textit{Misner-Sharp} (or \textit{Hawking}) \textit{mass}
\begin{equation*}
m = \frac{r}{2} \left(1 - \frac{\dot r^2}{r^2}\right).
\end{equation*}
If $\tau$ and $r$ are additionally even functions of $\rho$, then the wormhole is said to be \textit{symmetric}.
\end{definition}

\begin{remark*}
The asymptotic conditions on the metric coefficients ensure that the induced Riemannian metric on the spacelike slices $\{t = \text{const.}\}$ satisfies
\begin{equation*}
g_{ij} = \delta_{ij} + O(1/r)
\end{equation*} 
in the standard Cartesian coordinates $x^i$.
One could obtain decay rates for the derivatives of the metric $g_{ij}$ by imposing additional boundedness assumptions on the derivatives of the mass $m$, but we do not include this in the general definition. 
\end{remark*}

Note that the asymptotic flatness condition on $r$ implies that there exists \textit{at least one} point at which $r$ has a minimum.
The local minima (resp.\ maxima) of $r$ are usually referred to as wormhole \textit{throats} (resp.\ \textit{bellies}).
If $\rho=\rho_0$ describes a wormhole throat, then we have
\begin{equation*}
\dot r(\rho_0) = 0 \quad\text{and}\quad \ddot r(\rho_0) \geq 0.
\end{equation*}
In fact, one sometimes requires that the latter is strictly positive at a throat, in which case the condition is called the \textit{flare-out condition}, although it is strictly speaking not necessary (indeed, the main point is that a wormhole should describe a connection betwen two asymptotically flat universes).
It should however be noted that a throat with $\ddot r(\rho_0)=0$ is degenerate in the sense that it has vanishing surface gravity, cf.\ \cite{hayward-sph-sym-wh}.

Definition \ref{def-wh} describes a wormhole as a geometric object, but the ultimate goal is of course to obtain a spacetime model. Thus, one should aim to construct wormhole geometries satisfying the Einstein field equations for some accepted matter model.
In doing so, one must show some leniency in deciding which matter models which are to be considered acceptable, as wormholes require support from matter violating the null energy condition \cite{morris-thorne,visser,hayward-dynamic-wh}.
Matter of this type is often called \textit{exotic}, and is the main reason why wormholes are still considered to be hypothetical from a physics perspective.
Possibly the most natural way of modelling exotic matter is by means of a \textit{phantom field} (or \textit{ghost}), which is a scalar field with a reversed sign in front of its energy density in the Lagrangian (this will be clarified in the next subsection).
Such fields often appear in cosmological research, as they could explain the accelerated expansion of the universe \cite{phantom-cosmology}.
Matter models coupled to such a field will henceforth be referred to as \textit{haunted}.

\subsection{Haunted Einstein-Yang-Mills theory}
The haunted theory of gravity that we consider in this work also carries non-abelian electromagnetic charge in the form of a Yang-Mills field, and will be referred to as the \textit{haunted Einstein-Yang-Mills (EYM) theory}.
Let us briefly describe the setting.
Let $M$ be a smooth $n$-dimensional manifold, $G$ a compact Lie group with a bi-invariant metric $\langle\cdot,\cdot\rangle$, and $P$ a principal $G$-bundle over $M$.
We consider the \textit{haunted EYM Lagrangian}
\begin{align}\label{eq-eym-phantom-functional}
(g,\omega,\phi) \mapsto \left(\Scal_g - \Vert F_\omega \Vert^2 + \Vert \mathrm d\phi \Vert^2 \right) \vol_g,
\end{align}
where 
\begin{itemize}
\item $g$ is a semi-Riemannian metric on $M$ with scalar curvature $\Scal_g$ and volume form $\vol_g$, 
\item $\omega$ is a connection on $P$ with curvature two-form $F_\omega =  \tfrac12\,\mathrm dx^\mu \wedge \mathrm dx^\nu \otimes F_{\mu\nu}$,
\item $\phi:M\to\rn$ is a smooth function, called the \textit{phantom field},
\item the norms are given in local coordinates by
\begin{equation*}
\Vert F_\omega \Vert^2 = \tfrac{1}{2} \langle F_{\mu\nu}, F^{\mu\nu} \rangle
\quad\text{and}\quad 
\Vert \mathrm d\phi \Vert^2 = \partial_\mu \phi \, \partial^\mu \phi.
\end{equation*}
\end{itemize} 
Integrating the Lagrangian (\ref{eq-eym-phantom-functional}) over compact subsets of $M$ and varying with respect to $(g,\omega,\phi)$ leads to the (trace-reversed) Einstein field equation(s), the Yang-Mills equation, and the phantom field equation:
\begin{subequations}
\begin{empheq}[left=\empheqlbrace]{align}
\label{eq-efe-general}  \Ric_g - 2\langle F_\omega \otimes F_\omega \rangle + \tfrac1{n - 2} \Vert F_\omega \Vert^2g + \mathrm d\phi\otimes \mathrm d\phi &= 0,\\[0.1cm]
\label{eq-ym-general}  \mathrm d_\omega \star F_\omega &= 0,\\[0.1cm]
\label{eq-ph-general}  {\textstyle\dalembertian_g} \phi &= 0,
\end{empheq}
\end{subequations}
where 
\begin{itemize}
\item $\mathrm d_\omega$ denotes the covariant derivative with respect to the connection $\omega$, 
\item $\star$ is the Hodge star operator with respect to $g$, 
\item ${\textstyle\dalembertian_g} = \nabla^\mu\nabla_\mu$ is the wave operator with respect to $g$, 
\item the tensor field $\langle F_\omega \otimes F_\omega \rangle$ is given in a gauge and local coordinates by
\begin{equation*}
\langle F_\omega \otimes F_\omega \rangle = \frac{1}{2} \langle F_{\mu\alpha}, F\indices{_\nu^\alpha} \rangle \, \mathrm dx^\mu \otimes \mathrm dx^\nu.
\end{equation*}
Note that $\tr_g \langle F_\omega \otimes F_\omega \rangle = \Vert F_\omega \Vert^2$.
\end{itemize}
The procedure of deriving the equations is standard in literature, so we omit it for brevity.
The reader should perhaps only note the reversed sign in front of the scalar field term in (\ref{eq-efe-general}). 

If the connection $\omega$ is flat, so that $F_\omega\equiv 0$, then there is no contribution of the Yang-Mills field to the system and one obtains the haunted Einstein equations, which are solved e.g.\ by the Ellis-Bronnikov wormhole (\ref{eq-ellis-bronnikov}).
On the other hand, setting $\phi\equiv 0$ yields the classical EYM system.

\subsection{Static spherically symmetric $\SU 2$ ansatz}
We are interested in manufacturing wormhole solutions of the system (\ref{eq-efe-general}--\ref{eq-ph-general}).
Following Definition \ref{def-wh}, we therefore assume that our spacetime has the four-dimensional topology $M = \rn^2\times\mathbb S^2$ and is endowed with a Lorentzian metric $g$ of the form
\begin{equation*}
g = -e^{2\tau(\rho)} \, \mathrm dt \otimes \mathrm dt + r(\rho)^2 \left( \mathrm d\rho \otimes \mathrm d\rho + \gS \right).
\end{equation*}
We must also choose a principal $G$-bundle over $M$ and (more importantly) an ansatz for a connection on this bundle.
Here, we wish to assume some further symmetries in order to make the problem feasible.
To this end, observe that the symmetry group $K=\SU 2$ acts naturally on the spherical factor of $M$ by isometries.
A principal $G$-bundle $P \xrightarrow\pi M$ is said to be \textit{spherically symmetric} (or more generally \textit{$K$-symmetric}) if the action of the symmetry group $K=\SU 2$ on $M$ admits a lift to a left action on the total space $P$ by bundle automorphisms.
As far as the gauge group $G$ is concerned, one of the simplest non-trivial yet sufficiently rich (cf.\ next subsection) choices turns out to also be $G = \SU 2$, and we will work with this choice throughout this manuscript.
Now, using the general theory of $K$-symmetric principal $G$-bundles \cite{brodbeck-symm,harnad-bundle-1}, one can show that the equivalence classes of spherically symmetric principal $\SU 2$-bundles over $M$ are in one-to-one correspondence with non-negative integers $0 \leq m\in\mathbb Z$.
Using the same theory, one can also classify spherically symmetric (i.e.\ invariant under the left action of $K=\SU 2$) connections on these bundles.
In fact, in the only non-trivial case%
\footnote{For $m\not=1$, any spherically symmetric connection is necessarily $\mathfrak u(1)$-valued, so the $\SU 2$ Yang-Mills theory degenerates to the classical Maxwell theory.}
$m=1$, a gauge can be constructed so that a general spherically symmetric connection has the well-known \cite{bartnik-connection,breit-forg-mais,bart-mckin} form
\begin{equation*}
\omega = \, (w\,\mathrm d\theta - v \sin\theta\, \mathrm d\varphi)\otimes X
+ (v\,\mathrm d\theta + w\sin\theta\,\mathrm d\varphi) \otimes Y
+ (a\,\mathrm dt + b\,\mathrm d\rho + \cos\theta \, \mathrm d\varphi)\otimes Z,
\end{equation*}
where $(\theta,\varphi) \in \mathbb S^2$ are the standard spherical coordinates coordinates, $w,v,a,b$ are smooth functions of $(t,\rho)$, and $X,Y,Z$ is the standard basis for the Lie algebra $\su 2$ given by 
\begin{equation*}
X =- 
\frac{i}{2}
\begin{bmatrix}
0 & 1\\
1 & 0
\end{bmatrix},
\quad
Y =-
\frac{i}{2}
\begin{bmatrix}
0 & -i\\
i & 0
\end{bmatrix},
\quad
Z =
-\frac{i}{2}
\begin{bmatrix}
1 & 0\\
0 & -1
\end{bmatrix}.
\end{equation*}
Since we are working in the static setting, we will additionally assume that the functions $w,v,a,b$ depend only on $\rho$ and not on $t$.
This leads to the following simplifications:
\begin{itemize}
\item We may set $b\equiv 0$ by a gauge transformation depending only on $\rho$;
\item The Yang-Mills equation (\ref{eq-ym-general}) in this setting implies that $w$ and $v$ are constant multiples of one another so that we may also set $v \equiv 0$ by making a constant gauge transformation;
\item Similarly as in \cite{bizon-popp}, it can be shown \cite{thesis} that if $a$ is not identically zero and the metric $g$ is asymptotically flat with two ends as in Definition \ref{def-wh}, then the Yang-Mills equation (\ref{eq-ym-general}) necessitates $w \equiv 0$, in which case the haunted EYM system reduces to the haunted Einstein-Maxwell system and the solutions of the equation system can be obtained explicitly \cite{charged-wh}. We therefore do not lose generality by also setting $a \equiv 0$.  
\end{itemize}
Thus we work with the so-called \textit{purely magnetic} ansatz
\begin{equation}\label{eq-ym-connection}
\omega = w(\rho) \left( \mathrm d\theta \otimes X + \sin\theta\,\mathrm d\varphi \otimes Y \right) + \cos\theta\,\mathrm d\varphi \otimes Z,
\end{equation}
depending on a single function $w$.
The corresponding curvature form is given by
\begin{equation*}
F_\omega = \dot w\,\mathrm d\rho\wedge [\mathrm d\theta \otimes X + \sin\theta\,\mathrm d\varphi \otimes Y]
- (1-w^2)\sin\theta\,\mathrm d\theta\wedge\mathrm d\varphi \otimes Z.
\end{equation*}

Last but not least, we also assume that the phantom field $\phi$ is also static and spherically symmetric, so that it only depends on $\rho$.
For this setting, the phantom field equation (\ref{eq-ph-general}) admits the general solution formula
\begin{equation}\label{eq-ph-sol}
\phi(\rho) = \phi_0 + \pi_0\int_0^\rho \frac{1}{re^{\tau}}, \qquad \phi_0, \pi_0 \in \rn,
\end{equation}
which can then be inserted directly back into the Einstein field equation (\ref{eq-efe-general}) to eliminate $\phi$ from the system.

Putting everything together, we see that the haunted EYM system (\ref{eq-efe-general}--\ref{eq-ph-general}) is equivalent to
\begin{subequations}
\begin{empheq}[left=\empheqlbrace]{align}
\label{eq-ym}\displaystyle
\ddot w + \left(\dot\tau - \frac{\dot r}{r}\right)\dot w + w(1-w^2) &= 0,
\\[0.1cm]
\label{eq-efe-tau}\displaystyle
\ddot\tau + \dot\tau^2 + \frac{\dot r\dot\tau}{r} - \frac{2\dot w^2}{r^2} - \frac{(1-w^2)^2}{r^2} &= 0,
\\[0.1cm]
\label{eq-efe-r}\displaystyle 
\frac{\ddot r}{r} + \frac{\dot r\dot\tau}{r} - 1 + \frac{(1-w^2)^2}{r^2} &= 0,
\\[0.1cm]
\label{eq-efe-constraint}\displaystyle
1 + \frac{2\dot w^2}{r^2} - \frac{(1-w^2)^2}{r^2} - \frac{\dot r}{r} \left(\frac{\dot r}{r} + 2\dot\tau \right) &= \frac{\pi_0^2}{(re^\tau)^2},
\end{empheq}
\end{subequations}
where (\ref{eq-ym}) is the Yang-Mills equation and the remaining equations arise from the Einstein field equation (\ref{eq-efe-general}).
The free parameter $\pi_0$ can be interpreted as the charge of the phantom field (\ref{eq-ph-sol}), and setting $\pi_0 = 0$ yields the classical EYM theory.
In fact, one easily verifies that the equation (\ref{eq-efe-constraint}) is implied by the other three equations (\ref{eq-ym}--\ref{eq-efe-r}), assuming that it holds at least at one point.
Hence, we may view it as a constraint on the initial conditions.

\subsection{Past and present}

The classical $\SU 2$ EYM system, i.e.\ (\ref{eq-ym}--\ref{eq-efe-constraint}) with $\pi_0=0$, was extensively studied in the late 20th century, cf.\ \cite{volkov-galtsov}.
This was initiated by Bartnik and McKinnon \cite{bart-mckin} and Bizoń \cite{bizon-bh}, when they numerically found particle-like and black hole solutions to these equations.
It has since been mathematically shown that these equations in fact admit infinite sequences of particle-like and black hole solutions.
This was first done in the series of papers \cite{smol-wass-1,smol-wass-2,smol-wass-3} by Smoller, Wasserman et al.
A complete classification of the solutions to the equations was later provided in \cite{breit-forg-mais} by Breitenlohner, Forg\'acs and Maison, which also allowed for a somewhat more elegant existence proof.
In a later work \cite{yang-mills-dilaton}, Maison also performed a similar analysis of the $\SU 2$ Yang-Mills-dilaton system, which can interestingly be put into a similar form as the EYM system and in fact also allows for an infinite sequence of solutions, although the proof is more involved, despite the simpler appearance of the system.

As for the haunted $\SU 2$ EYM framework, i.e.\ with $\pi_0\not=0$ in (\ref{eq-ym}--\ref{eq-efe-constraint}), sequences of wormhole solutions have been constructed numerically by Hauser, Ibadov, Kleihaus, Kunz \cite{hairy-wormholes}, and the main purpose of the present manuscript is to mathematically prove the existence of these solutions.

\begin{theoremintro}\label{thm-sym-wh}
For each $r_0 > 0$, there exists a sequence 
$$\{(\tau^{(n)}, r^{(n)}, w^{(n)})\}_{n\geq0}$$
of symmetric wormhole solutions to the system {\normalfont (\ref{eq-ym}--\ref{eq-efe-constraint})}, satisfying $r^{(n)}(0) = r_0$. 
The Yang-Mills potential $w^{(n)}$ satisfies $|w^{(n)}| \to 1$ as $\rho\to\pm\infty$, has $n$ zeros, and is symmetric, i.e.\ even or odd with the same parity as $n$.
Furthermore:
\begin{itemize}
\item For $r_0 \geq 1$, the wormholes have a single throat at $\rho=0$.
\item For $r_0 < 1$ and even (resp.\ odd) $n$, the wormholes have a non-degenerate throat (resp.\ belly) at $\rho=0$.
In particular, the wormholes have at least two throats for odd $n$.
\end{itemize}
\end{theoremintro}

\begin{remark*}
By a wormhole solution of (\ref{eq-ym}--\ref{eq-efe-constraint}), we mean a solution such that the metric coefficients $\tau$ and $r$ meet the criteria of Definition \ref{def-wh}.
We also recall that symmetry of a wormhole here means that the metric coefficients $\tau$ and $r$ are even functions.
\end{remark*}

The wormhole solutions from Theorem \ref{thm-sym-wh} bear a lot of resemblance to the black hole solutions of the classical $\SU 2$ EYM equations discussed above.
In fact, the proof of Theorem \ref{thm-sym-wh} follows closely the blueprint laid forth in the already mentioned work \cite{breit-forg-mais}, which essentially boils down to using a shooting method to construct the desired sequences of solutions.
However, even though one might expect that this procedure requires only a simple modification of the already existing proofs for the classical EYM system, it turns out that the phantom field destroys certain properties that the classical EYM system has, which also makes the proofs more difficult.
In particular, one of the main difficulties here is that certain quantities lose their monotonicity properties and could even oscillate, which makes the behaviour of the solutions somewhat analogous to those of the aforementioned Yang-Mills-dilaton system \cite{yang-mills-dilaton}
The shooting method in our case also requires the development of certain new techniques, in particular in the proof of the existence of wormholes whose Yang-Mills potential has an odd number of zeros (the analogues of these solutions were not interesting in the context of classical EYM theory, and consequently were not studied).

\subsection{First-order initial value problem}
Before proceeding with a description of the proof of Theorem \ref{thm-sym-wh}, we rewrite (\ref{eq-ym}--\ref{eq-efe-constraint}) as a first-order system (following \cite[\S 6]{breit-forg-mais}) by defining the new dependent variables
\begin{equation}\label{eq-breit-variables}
N = \frac{\dot r}{r}, \quad U = \frac{\dot w}{r}, \quad \kappa = \dot\tau + N, \quad \zeta = \frac{\pi_0}{re^\tau}.
\end{equation}
Thus, the system (\ref{eq-ym}--\ref{eq-efe-constraint}) transforms to
\begin{subnumcases}{}
\dot r = rN, \label{eq-r}\\[0.1cm]
\dot N = 1 - \tfrac{(1-w^2)^2}{r^2} - \kappa N, \label{eq-N}\\[0.1cm]
\dot w = rU, \label{eq-w}\\[0.1cm]
\dot U = -(\kappa-N)U - \tfrac{w(1-w^2)}{r}, \label{eq-U}\\[0.1cm]
\dot \kappa = 1 + 2U^2 - \kappa^2, \label{eq-kappa} \\[0.1cm]
\dot \zeta = -\kappa \zeta,\label{eq-zeta}
\end{subnumcases}
together with the constraint (\ref{eq-efe-constraint})
\begin{equation}
\zeta^2 = 1 + 2U^2 - \frac{(1-w^2)^2}{r^2} - 2\kappa N + N^2. \label{eq-kappa-constraint}
\end{equation}
Using this constraint, we can also rewrite (\ref{eq-N}) as
\begin{equation}\label{eq-N-2}
\dot N = (\kappa-N)N - 2U^2 + \zeta^2 = \frac12 \left(1 - N^2 - 2U^2 - \frac{(1-w^2)^2}{r^2} + \zeta^2\right).
\end{equation}

\begin{remark*}
Even though the equation (\ref{eq-zeta}) for the phantom term $\zeta$ is decoupled from the rest of the system (\ref{eq-r}--\ref{eq-kappa}), $\zeta$ still appears in the constraint (\ref{eq-kappa-constraint}), and one uses this constraint repeatedly throughout the analysis.
E.g.\ we will often use the alternate forms (\ref{eq-N-2}) of the equation for $N$.
We therefore keep $\zeta$ as a dependent variable.
\end{remark*}

The system (\ref{eq-r}--\ref{eq-zeta}) is regular as long as $r>0$.
Hence, for any choice of initial conditions with $r(0)>0$, there exists a unique local (real) analytic solution to the problem by standard ODE theory.

Note that the radial function $r$ of a wormhole spacetime requires at least one point at which $r$ is stationary (cf.\ Definition \ref{def-wh}), so it is natural to assume the initial value $N(0) = 0$.
Note that
\begin{equation*}
\dot N(0) = 1 - \frac{(1-w(0)^2)^2}{r(0)^2}, 
\end{equation*}
so that if $r(0) + w(0)^2 > 1$, then the initial conditions describe a wormhole throat.
However, since the wormhole could have several throats, the stationary point could also describe a belly, so we do not enforce this condition.
In fact, we will mainly focus on symmetric wormholes, and in some cases they will be symmetric around a belly rather than a throat.
For the constraint (\ref{eq-kappa-constraint}) to be satisfied, we also need to assume
\begin{equation*}
\zeta(0)^2 = 1 + 2U(0)^2 - \frac{(1-w(0)^2)^2}{r(0)^2},
\end{equation*}
which can only be satisfied if the right hand is non-negative.
The initial value $\kappa(0)$ is not a priori constrained in any way, other than the requirement that it should be finite (note that, for a black hole horizon, one would need $\kappa(0)=\infty$ \cite[\S 6]{breit-forg-mais}).
However, the analysis of the equations is considerably simplified by making the assumption $\kappa(0) = 0$, which we will do throughout the manuscript.
Finally, the initial values $w(0)$ and $U(0)$ for the Yang-Mills potential are allowed to be arbitrary, but we assume for simplicity that $|w(0)|\leq 1$ because the solutions that exit the strip $|w|\leq 1$ will turn out to be ill-behaved.

Thus, we supplement the system (\ref{eq-r}--\ref{eq-zeta}) with the initial conditions
\begin{equation}\label{eq-initial-vals}
\begin{array}{lll}
r(0) = r_0, & w(0) = w_0, & \kappa(0) = 0, \\[0.2cm]
N(0) = 0, & U(0) = U_0, & \zeta(0) = \sqrt{E_0},
\end{array}
\end{equation}
such that the parameters $(r_0, w_0, U_0)$ belong to the \textit{set of admissible initial data}
\begin{equation}\label{eq-adm-vals}
\mathscr{I}_0 = \left\{ (r_0, w_0, U_0) \in \mathbb R^3 \;\Big|\; r_0 > 0, \; |w_0|\leq1, \; E_0 \geq 0\right\},
\end{equation}
where we denote by
\begin{equation*}
E_0 = 1 + 2U_0^2 - \frac{(1-w_0^2)^2}{r_0^2}
\end{equation*}
the initial value of the \textit{energy}, a quantity which will turn out to have useful properties.

\begin{remark}\label{remark-continuity-initial-data}
This particular choice of initial conditions depends continuously (note the square root in the initial condition for $\zeta$) on the initial data $(r_0,w_0,U_0) \in \mathscr I_0$, so the solutions of the initial value problem also depend continuously on the initial data.
\end{remark}

\begin{remark}\label{remark-tau0-free}
Since we will eventually want to transform back to the original variables $(\tau, r, w)$, we finally wish to recall here that the initial value $\tau_0$ of the temporal metric coefficient $\tau$ is kept free and instead we fix the phantom charge
\begin{equation}\label{eq-phantom-charge}
\pi_0 = r_0e^{\tau_0}\sqrt{E_0},
\end{equation} 
in order to remain compatible with $\zeta = \pi_0/(re^\tau)$.
\end{remark}

\begin{remark}\label{remark-symmetries}
Note that the equations (\ref{eq-r}--\ref{eq-zeta}) possess the symmetries
\begin{equation*}
(w,U) \mapsto -(w,U) \quad\text{and}\quad (\rho, N, U, \kappa) \mapsto -(\rho, N, U, \kappa). 
\end{equation*}
Thus, the general solution to the initial value problem with initial data $(r_0,w_0,U_0) \in \mathscr I_0$ satisfies the identities
\begin{align*}
& \;  (r,\phantom{-}N,\phantom{-}w,\phantom{-}U,\phantom{-}\kappa,\phantom{-}\zeta)(\phantom{-}\rho,\phantom{-}r_0,\phantom{-}w_0,\phantom{-}U_0) \\[0.1cm]
=& \; (r,\phantom{-}N,-w,-U,\phantom{-}\kappa,\phantom{-}\zeta)(\phantom{-}\rho,\phantom{-}r_0,-w_0,-U_0) \\[0.1cm]
=& \; (r,-N,\phantom{-}w,-U,-\kappa,\phantom{-}\zeta)(-\rho,\phantom{-}r_0,\phantom{-}w_0,-U_0)\\[0.1cm]
=& \; (r,-N,-w,	\phantom{-}U,-\kappa,\phantom{-}\zeta)(-\rho,\phantom{-}r_0,-w_0,\phantom{-}U_0).
\end{align*}
Solutions with either $w_0=0$ or $U_0=0$ are therefore symmetric, since in that case $r,\zeta$ are even functions, while $N,\kappa$ are odd.
Moreover:
\begin{itemize}
\item  If $U_0=0$, then $w$ is even, so these are often referred to as \textit{even} solutions.
\item  If $w_0=0$, then $w$ is odd, so these are often called \textit{odd} solutions.
\end{itemize}
In particular, we see here that $w_0=U_0=0$ implies $w\equiv U \equiv 0$.
With these symmetries in mind, we see that it suffices to study the solutions for $\rho \geq 0$.
\end{remark}

\subsection{Proof summary and organization}

To prove Theorem \ref{thm-sym-wh}, we wish to show the existence of initial data of the form $(r_0,w_0,0)$ and $(r_0,0,U_0)$, for which the solution of (\ref{eq-r}--\ref{eq-zeta},\ref{eq-initial-vals}) is defined for all $\rho \geq 0$ and the dependent variables have the following limits as $\rho \to \infty$:
\begin{equation*}
r \to \infty, \quad  r(1-N^2) \to \beta, \quad r\zeta \to \alpha, \quad |w|\to 1, 
\end{equation*}
where 
\begin{equation*}
0 < \alpha = \pi_0e^{-\tau_\infty} < \infty \quad\text{and}\quad 0 \leq \beta = 2m_\infty < \infty.
\end{equation*}
The proof is based on a shooting method, in the sense that we consider infima over suitably chosen subsets of the set of initial data $\mathscr I_0$, with the expectation that these infima correspond to initial data describing the desired global solutions.
The proof consists of three main ingredients:

\begin{itemize}
\item \textit{Classification theorem} (\S \ref{sec-classification}): It turns out that, just as in the setting of \cite{breit-forg-mais}, the initial value problem admits three possible types of solutions, classified according to the behaviour of $N$.
In fact, $N$ plays a more important role than the other dependent variables, since it will turn out the solution can only stop existing if $N \to -\infty$, and even in that case the remaining dependent variables stay bounded. 
Of particular importance will be the dichotomy between the regions $N+\zeta < 0$ and $N+\zeta \geq 0$, which also demonstrates that the phantom field $\zeta$ affects the solutions in a non-trivial way.
Indeed, any orbit entering the former region will turn out to be singular, whereas the orbits staying in the latter region will be well-defined for all $\rho\geq0$.
The strip $|w|\leq 1$ will also play a major role, and any orbit exiting it will also turn out to be singular.
These facts will allow us to work in the region $\{N+\zeta \geq 0,\,|w|\leq 1\}$, in which the solutions are generally well-behaved.
We will then study the asymptotic behaviour of the dependent variables, depending on the behaviour of $N$ (in particular the number of its zeros and its sign near infinity), and show that only a handful of cases can occur.
One of the biggest difficulties will be the fact that, prima facie, we do not know whether the dependent variables even have limits at infinity, so that various techniques will be applied to extract these limits (note also Remark \ref{remark-bfm-gap} below).
An important tool for this will be a simple result called \textit{Barbălat's lemma}.

\begin{lemma}[Barbălat's lemma \cite{barbalat}]\label{lemma-barbalat}
Suppose that $f:[0,\infty) \to \rn$ is a uniformly continuous function such that
\begin{equation*}
\lim_{y \to \infty} \int_0^y f(x)\,\mathrm dx \quad \text{exists and is finite.}
\end{equation*}
Then $f(x) \to 0$ as $x\to\infty$.
\end{lemma}
We recall that a simple sufficient condition for the uniform continuity of $f$ is that its derivative is bounded (provided that $f$ is differentiable).
\item \textit{Neighbourhood theorem} (\S \ref{sec-neighbourhood-theorem}): For each solution type, we will study the solutions with nearby initial data.
Such a result will provide us with a mechanism of controlling which solution types the aforementioned infima can describe.
The proof is mainly based on the fact that the solutions depend continuously on the initial data, but certain cases require a closer analysis (again, note also Remark \ref{remark-bfm-gap} below).
\item \textit{Energy theorem} (\S \ref{sec-energy-theorem}): We study solutions with small resp.\ large initial energy $E_0$. 
This will provide us with an upper and lower bound for the shooting method.
The small case $E_0=0$ will follow by a simple analysis of the initial conditions.
The case when $E_0$ is large will, however, be much more involved.
In fact, this case has not been studied in the aforementioned citations, so the proof given here involves completely new techniques developed specifically for the problem corresponding to wormhole initial conditions. 
\end{itemize}

With these three results in hand, Theorem \ref{thm-sym-wh} will follow by a simple shooting method argument, the details of which we leave for \S \ref{sec-construction}.

\begin{remark}\label{remark-bfm-gap}
We would like to point out a potential gap in some of the the proofs given in \cite{breit-forg-mais} for the analogues of the above mentioned theorems.
Therein, the authors study the asymptotic behaviour of the dependent variables using some heavy machinery of dynamical systems, namely the theory of structurally stable vector fields \cite[\S 1.3]{anosov-arnold-dynamical} (particularly in their analogues of the classification and the neighbourhood theorem).
This is done by viewing the Yang-Mills equation
$$ \ddot w + (\kappa-2N)\dot w + w(1-w^2) = 0 $$
as a perturbation of the Yang-Mills equation in the flat limit ($\kappa\to 1, \, N\to1$) and the cylindrical limit ($\kappa\to1,\, N \to 0$) i.e.,
$$ \ddot w \pm \dot w + w(1-w^2) = 0. $$
The latter equation is the so-called \textit{Duffing type equation}, which can be studied using the elementary theory of planar autonomous ODE systems.
While this certainly provides a good heuristic overview of how the solutions should behave, the author of the present manuscript does not understand how one can mathematically apply the theory of structural stability in this context.
Aside from some technical difficulties such as the fact that the vector field corresponding to the equation is tangential to (at least some points of) the boundary of any compact set containing the equilibria (whereas the theory assumes transversality at the boundary), the entire theory of structural stability only applies to \textit{autonomous} perturbations of autonomous planar dynamical systems.
On the other hand, the idea here is to consider $(\kappa-2N+1)\dot w$ as a small (for large $\rho$) perturbation of the flat Yang-Mills equation, with $\kappa$ and $N$ being interpreted as fixed externally given functions.
But such a perturbation is clearly non-autonomous, so that the theory of structural stability cannot be applied directly.
The author was also unable to find other references containing results that could be applied in this context.
In view of this, we take on a more raw analytical approach in the present work, which arguably also simplifies the proofs.
We would like to point out that the methods used in this work can also be applied in the context of the above mentioned particle-like and black hole settings, if this is indeed deemed necessary.
\end{remark}

\section{Classification of solutions}
\label{sec-classification}
% !TEX root = ms.tex

The goal of this section is to show the following classification result.

\begin{theorem}[Classification theorem]\label{thm-classification}
Any solution of the initial value problem {\upshape (\ref{eq-r}--\ref{eq-zeta}, \ref{eq-initial-vals})} with respect to fixed initial data in $\mathscr I_0$ belongs to one of the following classes:
\begin{enumerate}
\item[(i)] Singular: There exists a finite point $\rho_\infty > 0$ such that
\begin{equation*}
r \to 0, \quad N \to -\infty \quad\text{as}\quad \rho\to \rho_\infty,
\end{equation*}
and the remaining dependent variables remain bounded as $\rho \to \rho_\infty$.
\item[(ii)] Asymptotically cylindrical: The solution is defined for all $\rho \geq 0$, stays in the region $|w|\leq 1$, and the dependent variables have the following limits at infinity:
\begin{equation*}
r\to 1, \quad N \to 0, \quad \zeta\to 0, \quad \kappa \to 1.
\end{equation*}
Furthermore, either
\vskip0.1cm
\begin{itemize}
\item $r \equiv 1$ and $w \equiv 0$, or
\item $r_0 < 1, \; (w,U) \to (0,0)$ as $\rho \to\infty$, and $w$ has infinitely many zeros.
\end{itemize}
\item[(iii)] Asymptotically flat: The solution is defined for all $\rho \geq 0$, stays in the region $|w|\leq 1$, and the dependent variables have the following limits at infinity:
\begin{equation*}
r\to \infty, \quad r(1-N^2) \to \beta, \quad r\zeta \to \alpha, \quad \kappa \to 1,
\end{equation*}
where $0 < \alpha < \infty$ and $0\leq \beta < \infty$.
Furthermore, either
\vskip0.1cm
\begin{itemize}
\item $r_0 > 1$ and $w \equiv 0$, or
\item  $(w,rU) \to (\pm 1,0)$ as $\rho \to \infty$.
\end{itemize}
\end{enumerate}
\end{theorem}

\begin{remark*}
Note that $r\zeta = \pi_0e^{-\tau}$, so that $\tau$ has a finite limit at infinity in case (iii) provided that $\pi_0\not=0$.
By the constraint (\ref{eq-kappa-constraint}), this limit can be calculated as
\begin{equation*}
\tau_\infty = -\frac12\log\left[\frac{1}{\pi_0^2}\lim_{\rho\to\infty} r^2(1-2\kappa N + N^2)\right],
\end{equation*}
but does not seem to admit a closed form in terms of the initial conditions.
\end{remark*}

This classification is highly reminiscent of the one given in \cite[Theorem 16]{breit-forg-mais},
where the Einstein-Yang-Mills equations (with no phantom field) are studied for particle-like and black hole initial conditions.
The proof in our context is, however, more involved in view of the increased complexity of the behaviour of $N$.
In fact, the main feature of the phantom system (as opposed to the phantomless one) is that $N=\dot r/r$ is allowed to change sign without the orbit being singular.

%%%%%%%%%%%%%%%%%%%%%%%%%%

\subsection{Trivial solutions}\label{subsec-trivial-solutions}

Note that for $w_0 \in \{-1, 0, 1\}$ and $U_0 = 0$, we have that $w$ is identically constant $w \equiv w_0$ and hence also $U \equiv 0$. 
In this case we can explicitly solve
\begin{equation*}
\kappa(\rho) = \tanh(\rho),\qquad
\zeta(\rho) = \sqrt{E_0}\sech(\rho).
\end{equation*}
The remaining non-trivial equation is the Riccati type equation
\begin{equation*}%\label{eq-N-ric}
\dot N =  E_0\sech^2(\rho) + \tanh(\rho)N - N^2 = 1 - \frac{(1-w_0^2)^2}{r^2} - \tanh(\rho)N.
\end{equation*}
If $(w_0, U_0) = (\pm 1, 0)$, then $w\equiv \pm 1$ and we can also get the explicit solutions
\begin{equation*}
r(\rho) = r_0\cosh(\rho), \quad N(\rho) = \tanh(\rho),
\end{equation*}
for any $r_0 > 0$, corresponding to the Ellis-Bronnikov family of wormholes (\ref{eq-ellis-bronnikov}).

For solutions with $w_0 = U_0 = 0$, we have $w\equiv 0$.
In this case, the Yang-Mills connection (\ref{eq-ym-connection}) is $\mathfrak u(1)$-valued and the $\SU 2$ Yang-Mills theory degenerates to the classical Maxwell theory with gauge group $\mathbf U(1)$. 
These solutions can be computed explicitly as
\begin{equation*}
r(\rho) = r_0 \cosh(\rho) \cos\left( \frac{1}{r_0} \arctan(\sinh\rho) \right).
\end{equation*}
Note that the necessary assumption $E_0 \geq 0$ implies that $r_0 \geq 1$.
There are two separate cases:
\begin{itemize}
\item If $r_0 = 1$, then $r \equiv 1$ and thus $N\equiv 0$. This solution is asymptotically cylindrical.
\item If $r_0 > 1$, then the solutions are asymptotically flat.
They represent the abelian wormhole family, which has been obtained in \cite{charged-wh} as a solution of the haunted Einstein-Maxwell system.
\end{itemize}

%%%%%% PROOF START %%%%%%%%%%

\subsection{Proof of the classification} \label{subsec-classification-proof}

To simplify the statements of certain results, we say that a region $U \subset \rn^6$ in the phase space is \textit{(forward) invariant} if it has the following property: if there is a point $\rho_0\geq0$ such that the solution enters $U$ at $\rho=\rho_0$, then it stays in $U$ for all $\rho\geq\rho_0$, i.e.
$$ [\exists \rho_0 \geq 0 \,\colon\, (r,N,w,U,\kappa,\zeta)(\rho_0) \in U] \quad\Rightarrow\quad [\forall \rho\geq\rho_0, \,(r,N,w,U,\kappa,\zeta)(\rho) \in U]. $$
For a trivial example, we see from (\ref{eq-zeta}) that the region $\zeta > 0$ is invariant.

Throughout the rest of the manuscript, we will make extensive use of certain energy functions related to the equations.
In view of this, they deserve a proper definition.

\begin{definition}\label{def-energy}
The \textit{energy} of the system (\ref{eq-w}--\ref{eq-zeta}) is defined as the function
\begin{equation}\label{eq-energy}
E = 1 + 2U^2 - \frac{(1-w^2)^2}{r^2} = 2\kappa N - N^2 + \zeta^2,
\end{equation}
where the second equality follows from the constraint (\ref{eq-kappa-constraint}).
The \textit{autonomous energy} is defined as
\begin{equation}\label{eq-aut-energy}
F = 2\dot w^2 - (1-w^2)^2 = r^2(E-1).
\end{equation}
\end{definition}

\vskip0.1cm

We first derive some basic inequalities.

\begin{lemma}\label{lemma-ineq}
Consider a solution of the initial value problem {\upshape (\ref{eq-r}--\ref{eq-zeta}, \ref{eq-initial-vals})} with respect to fixed initial data in $\mathscr I_0$.
For all $\rho \geq 0$ for which the solution is defined, we have 
\begin{equation*}
\kappa \geq \tanh(\rho) \geq N,
\qquad
0 \leq \zeta \leq \sqrt{E_0}\sech(\rho),
\qquad 
\kappa+N \leq 2+\sqrt{E_0}\sech(\rho).
\end{equation*}
\end{lemma}

\begin{remark*}
The first set of inequalities implies that the temporal metric coefficient $\tau$ is non-decreasing, since $\kappa-N = \dot\tau$. 
Looking at the last two inequalities, one might hope that the stronger inequality $\kappa + N \leq 2 + \zeta$ holds, but this is in fact not true.
Indeed, numerical approximations suggest that this inequality is violated when the initial energy $E_0$ is large.
\end{remark*}

\begin{proof}
To prove the first inequality, let $\xi = \frac{1-\kappa}{1+\kappa}$ and calculate
\begin{equation*}
\dot\xi = -2\xi -U^2(1+\xi)^2 \leq -2\xi
\end{equation*}
Thus $e^{2\rho}\xi(\rho)$ does not increase, and $\xi(0)=1$ therefore implies $\xi(\rho) \leq e^{-2\rho}$, which can be rearranged to get
\begin{equation*}
\kappa \geq \frac{1-e^{-2\rho}}{1 + e^{-2\rho}} = \tanh(\rho).
\end{equation*}
The inequalities for $\zeta$ then follows easily because (\ref{eq-zeta}) implies
\begin{equation*}
\zeta = \zeta(0)\exp\left(-\int_0^\rho \kappa \right) \leq \sqrt{E_0}\exp\left(-\int_0^\rho \tanh(\rho) \right) = \sqrt{E_0}\sech(\rho).
\end{equation*}
Next, if $\nu = N-\tanh(\rho)$, then
\begin{equation*}
\dot \nu = - \frac{(1-w^2)^2}{r^2} - \kappa N + \tanh^2(\rho) \leq -\kappa \nu - [\kappa-\tanh(\rho)]\tanh(\rho) \leq -\kappa \nu,
\end{equation*}
so that $\nu$ decreases in the region $\nu>0$, and $\nu(0) = 0$ thus implies $\nu\leq 0$.
Finally, we set $$\eta = \kappa+N - 2- \sqrt{E_0}\sech(\rho)$$ and calculate
\begin{equation*}
\dot\eta = 1 + \zeta^2  - \frac14(\kappa+N)^2 - \frac34(\kappa-N)^2 + \sqrt{E_0}\tanh(\rho)\sech(\rho),
\end{equation*}
cf.\ \cite[Lemma 10]{breit-forg-mais}.
In the region $\eta \geq 0$, we have
\begin{equation*}
\kappa+N \geq 2+\sqrt{E_0}\sech(\rho),\qquad
\kappa-N \geq 2-2N+ \sqrt{E_0}\sech(\rho) \geq \sqrt{E_0}\sech(\rho),
\end{equation*}
where we use the fact that $N\leq\tanh(\rho) \leq 1$.
Hence, we get
\begin{equation*}
\dot \eta \leq -\sqrt{E_0}\sech(\rho)[1-\tanh(\rho)] < 0,
\end{equation*}
so that $\eta$ decreases in the region $\eta \geq 0$, which yields the desired inequality since $\eta(0) < 0$.
\end{proof}

\begin{lemma}\label{lemma-kappa-to-1}
Consider a solution of the initial value problem {\upshape (\ref{eq-r}--\ref{eq-zeta}, \ref{eq-initial-vals})} with respect to fixed initial data in $\mathscr I_0$.
If the solution is defined for all $\rho \geq 0$, then
\begin{equation*}
\liminf_{\rho\to\infty} \kappa \geq 1 \quad\text{and}\quad \zeta \to 0.
\end{equation*} 
Furthermore, if $U\to 0$, then $\kappa \to 1$.
\end{lemma}

\begin{proof}
The first part of the statement follows trivially from the inequalities in Lemma \ref{lemma-ineq}.
For the last claim, let $\varepsilon > 0$, define $a_\varepsilon = \sqrt{1+\varepsilon}$ and $\xi_\varepsilon = \frac{a_\varepsilon-\kappa}{a_\varepsilon + \kappa}$.
A simple calculation yields
\begin{equation*}
\dot\xi_\varepsilon = -2a_\varepsilon \xi_\varepsilon - (U^2-\varepsilon)(1+\xi_\varepsilon)^2.
\end{equation*}
If $U \to 0$, then $U^2 \leq \varepsilon$ and hence $\dot\xi_\varepsilon \geq -2a_\varepsilon\xi_\varepsilon$ for large $\rho$, which implies that $\limsup \kappa \leq a_\varepsilon$ for all $\varepsilon>0$ and letting $\varepsilon \to 0$ shows that $\limsup \kappa \leq 1$, giving $\kappa \to 1$. 
\end{proof}

Next, we show that all the dependent variables behave well as long as $N$ is bounded.

\begin{lemma}\label{lemma-N-bounded-finite}
Consider a solution of the initial value problem {\upshape (\ref{eq-r}--\ref{eq-zeta}, \ref{eq-initial-vals})} with respect to fixed initial data in $\mathscr I_0$.
Suppose that the solution is defined (at least) for $0 < \rho < \bar\rho$, and that $N$ remains bounded as $\rho\to\bar\rho$.
Then all the dependent variables remain bounded as $\rho\to\bar\rho$ and the solution can be continued across $\bar\rho$.
\end{lemma}

\begin{remark*}
In \cite{breit-forg-mais}, the authors show an analogue of Lemma \ref{lemma-N-bounded-finite} for their setting when $N$ is lower bounded by a positive constant \cite[Proposition 9]{breit-forg-mais} and when $N$ is negative \cite[Proposition 13]{breit-forg-mais}, but they do not show the result when $N$ approaches $0$ from above at $\bar\rho$.
This last case was likely just forgotten, although strictly speaking it is also necessary for their setting.
\end{remark*}

\begin{proof}
Since $N$ is bounded, it follows that $r$ is also bounded by (\ref{eq-r}).
From Lemma \ref{lemma-ineq}, we directly see that $\zeta$ is bounded, and also
\begin{equation*}
0 \leq \kappa \leq 2+\sqrt{E_0}-N,
\end{equation*}
so $\kappa$ is bounded as well.
So it remains only to study $w$ and $U$.
For this we will use the energies from Definition \ref{def-energy}.
Note that the energy $E=2\kappa N-N^2+\zeta^2$ is bounded and hence the autonomous energy $F$ is bounded as well.
In particular, we see that $w$ is bounded if and only if $U$ is bounded, so it suffices to show that $w$ is bounded at $\bar\rho$.

Aiming to reach a contradiction, assume that $w$ is unbounded at $\bar\rho$.
We first note that, since $w$ is unbounded, it must enter the region $|w|>1$ for some $0 \leq \rho_0 < \bar\rho$.
In view of the symmetry $(w,U)\mapsto -(w,U)$ (cf.\ Remark \ref{remark-symmetries}), we can without loss of generality assume that $w>1$.
In this region, $w$ is monotone since
\begin{equation*}
\ddot w = -(\kappa-2N) \dot w - w(1-w^2) > -(\kappa-2N) \dot w,
\end{equation*}
so $\ddot w|_{\dot w =0} > 0$, and hence $\dot w > 0$ for $w > 1$.
Thus $|w|\to \infty$, and hence by the Cauchy-Schwarz inequality
\begin{equation*}
\sqrt{\int_{\bar\rho-\delta}^{\rho} \dot w^2} 
\geq \frac{1}{\sqrt{\rho-\bar\rho + \delta}} \int_{\rho-\delta}^{\bar\rho} |\dot w|
\geq \frac{|w(\rho) - w(\bar\rho-\delta)|}{\sqrt{\rho-\bar\rho + \delta}} \to \infty,
\end{equation*}
as $\rho \to \bar\rho$, which implies that $\dot w$ is not square integrable near $\bar\rho$.
We will reach the desired contradiction by showing that $\dot w$ is, in fact, square integrable near $\bar\rho$.

Note that $\dot N \to -\infty$ by (\ref{eq-N}) since $\kappa, N, r$ are bounded and $|w|\to\infty$.
It follows that $N$ decreases for $\bar\rho-\delta \leq \rho \leq \bar\rho$ if $\delta > 0$ is sufficiently small, and in particular the limit
\begin{equation*}
\lim_{\rho\to\bar\rho} N = \bar N
\end{equation*}
exists.
We now consider two cases separately:
\begin{enumerate}
\item If $\bar N < \tfrac12\tanh(\bar\rho)$, then we can choose $\delta > 0$ so small that $2N - \tanh(\rho) \leq -c$ for $\bar\rho-\delta \leq \rho \leq \bar\rho$ and some positive constant $c$.
By Lemma \ref{lemma-ineq} (i), we then have
\begin{equation*}
\kappa - 2N \geq \tanh(\rho) - 2N \geq c \quad\text{for}\quad \bar\rho-\delta \leq \rho \leq \bar\rho,
\end{equation*}
and so
\begin{equation*}
\dot F = -4(\kappa-2N)\dot w^2 \leq -4c \dot w^2,
\end{equation*}
which implies that $\dot w$ is square integrable over $\bar\rho-\delta \leq \rho \leq \bar\rho$ since $F$ is bounded.
\item If $\bar N \geq \tfrac12 \tanh(\bar\rho) > 0$, then we can choose $\delta > 0$ so small that $N \geq \tfrac14 \tanh(\bar\rho) =: c$ for $\bar\rho-\delta \leq \rho \leq \bar\rho$.
If we define $b = r^2(1-N^2) = 2rm$, then $b$ is bounded since $r$ and $N$ are, and we have
\begin{equation*}
\dot b = 2r^2N \left( \frac{(1-w^2)^2}{r^2} + (\kappa-N)N \right) \geq 2\varepsilon(1-w^2)^2 = 2c (2\dot w^2 - F) \geq 4c \dot w^2 - \tilde c
\end{equation*}
for some constant $\tilde c > 0$, since $\kappa - N \geq 0$ by Lemma \ref{lemma-ineq} and $F$ is bounded.
This implies that $\dot w$ is square integrable over $\bar\rho-\delta \leq \rho \leq \bar\rho$.
\end{enumerate}
Thus, $\dot w$ is square integrable near $\bar\rho$ in both cases, which gives the desired contradiction.
\end{proof}

On the other hand, the following result characterizes singular orbits.

\begin{lemma}\label{lemma-sing}
Consider a solution of the initial value problem {\upshape (\ref{eq-r}--\ref{eq-zeta}, \ref{eq-initial-vals})} with respect to fixed initial data in $\mathscr I_0$.
\begin{enumerate}
\item The region $\{|w|>1,\, w\dot w > 0\}$ is invariant and any solution that enters it also enters the region $N+\zeta <0$.
\item The region $N+\zeta < 0$ is invariant and any solution that enters it is singular, cf.\ Theorem \ref{thm-classification} {\normalfont (i)}.
\end{enumerate}
\end{lemma}

\begin{remark*}
Note that, contrarily, a solution that enters the region $N+\zeta < 0$ does not necessarily also enter $\{|w|>1,\, w\dot w > 0\}$.
\end{remark*}

\begin{proof}
Suppose that the orbit enters the region $\{|w|>1,\, w\dot w > 0\}$.
In view of the symmetry $(w,U)\mapsto -(w,U)$ (cf.\ Remark \ref{remark-symmetries}), we can assume without loss of generality that there is a point $\rho_0\geq0$ with $w(\rho_0) > 1$ and $\dot w(\rho_0) > 0$.
Note that for $w > 1$,
\begin{equation*}
\ddot w = -(\kappa-2N) \dot w - w(1-w^2) > -(\kappa-2N) \dot w,
\end{equation*}
so $\ddot w|_{\dot w =0} > 0$, which shows that $w$ must keep increasing and hence the region $\{w > 1, \dot w > 0\}$ is invariant, so $(w,\dot w)$ remains there for all $\rho \geq \rho_0$.
Next, we want to show that the orbit enters $N + \zeta < 0$, so we study the orbit while it resides in the region $N + \zeta \geq 0$.
This implies that $1 \geq \tanh(\rho) \geq N \geq -\zeta \geq -\sqrt{E_0}\sech(\rho)$ by Lemma \ref{lemma-ineq}.
Note that the orbit exists as long as it stays in this region in view of Lemma \ref{lemma-N-bounded-finite}.
By Lemma \ref{lemma-ineq}, we have
\begin{equation*}
\kappa \leq 2 + \sqrt{E_0}\sech(\rho) - N \leq 2 + 2\sqrt{E_0}\sech(\rho).
\end{equation*}
Put $T = (w^2-1)/r > 0$ and calculate (cf.\ \cite[Proposition 11]{breit-forg-mais})
\begin{align*}
\frac{\mathrm d}{\mathrm d\rho} \log|TU| &= w\,\frac{2U^2 + T^2}{TU} - \kappa
\geq 2\sqrt{2} - \kappa
\geq 2\left[\sqrt 2 - 1 - \sqrt{E_0}\sech(\rho)\right].
\end{align*}
Thus, $|TU|$ increases strictly and uniformly for sufficiently large $\rho$, so the solution eventually reaches (and stays in) the region $|TU| \geq 1/\sqrt 2$, implying also that $2U^2 + T^2 \geq 2$. In this region, (\ref{eq-N-2}) gives
\begin{equation*}
\dot N = \frac12 (1 - N^2 - 2U^2 - T^2 + \zeta^2) \leq -\frac12[1 - E_0\sech^2(\rho)],
\end{equation*}
so $N$ uniformly decreases for large $\rho$, and thus it eventually reaches the region $N+\zeta < 0$ (recall from Lemma \ref{lemma-ineq} that $\zeta$ approaches zero), proving (i).

Now to prove (ii), put $\xi=N+\zeta$ and calculate
\begin{equation*}
\dot \xi = -\xi^2 + (k+2\zeta)\xi - 2\kappa\zeta - 2U^2.
\end{equation*}
Since $\kappa,\zeta \geq 0$, we see that $\xi$ decreases in the region $\xi < 0$ and consequently this region is preserved once reached.
In particular, if the orbit enters it, we have $\dot \xi \leq -\xi^2$, implying that $\xi \to -\infty$ at some finite point $\rho=\rho_\infty$,
which in turn implies that $N\to-\infty$ as $\zeta$ is bounded (note that none of the other variables can explode before $N\to-\infty$ in view of Lemma \ref{lemma-N-bounded-finite}).

Finally, we show that the other dependent variables remain bounded near the singular point $\rho_\infty$.
This will also imply that $r \to 0$ at $\rho_\infty$ e.g.\ by (\ref{eq-N}).
To this end, we adapt the techniques from \cite[Proposition 13]{breit-forg-mais}.
In fact, we only present the proof of the boundedness of $w$, as the boundedness of other variables follows in essentially the same way as in the citation, with only minor modifications.

Since $w$ is trivially bounded if it remains in the strip $|w|\leq 1$ for all $\rho < \rho_\infty$, we consider only the case when $w$ enters the invariant region $\{|w|>1,\, w\dot w > 0\}$,
and, as above, we assume without loss of generality that $w > 1, \, \dot w > 0$.
Put $\eta = -r(N+\zeta) = -r\xi$.
For sufficiently small $\delta$, we have $N + \zeta < 0$ and thus $\eta > 0$ for $\rho_\infty - \delta < \rho < \rho_\infty$.
We will show that $w\eta^{-\varepsilon}$ is bounded near $\rho_\infty$ for $0 < \varepsilon < \frac12$.
This will imply that $w$ is bounded because the constraint (\ref{eq-kappa-constraint}) gives
\begin{equation*}
\eta^2 = r^2(N^2-\zeta^2) - 2r\zeta\eta
\leq -r^2(1 + 2U^2 - 2\kappa N)  + (1-w^2)^2
\leq (1-w^2)^2,
\end{equation*}
so that $w\eta^{-\varepsilon}\geq w|1-w^2|^{-\varepsilon}$.
A simple calculation yields
\begin{equation*}
\dot \eta =  2rU^2 + \eta\zeta + r\kappa(\zeta-N)  > 0,
\end{equation*}
so that $\eta$ increases and in particular stays away from zero near $\rho_\infty$.
We have
\begin{equation*}
w\eta^{-\varepsilon}(\rho) - w\eta^{-\varepsilon}(\rho_\infty-\delta)
= \int_{\rho_\infty-\delta}^\rho \frac{\mathrm d}{\mathrm d\rho} (w\eta^{-\varepsilon})
= \int_{\rho_\infty-\delta}^\rho \dot w \eta^{-\varepsilon} - \int_{\rho_\infty-\delta}^\rho w \eta^{-\varepsilon} \dot\eta
\leq \int_{\rho_\infty-\delta}^\rho \dot w \eta^{-\varepsilon},
\end{equation*}
where the last inequality follows since $w > 1$ and $\dot\eta > 0$.
Now by the Cauchy-Schwarz inequality, we get
\begin{equation*}
\left(\int_{\rho_\infty-\delta}^\rho \dot w \eta^{-\varepsilon} \right)^2 \leq \int_{\rho_\infty-\delta}^\rho rU^2\eta^{-1-\varepsilon} \int_{\rho_\infty-\delta}^\rho r\eta^{1-\varepsilon},
\end{equation*}
so it suffices to show that the two integrals on the right-hand side are finite as $\rho\to\rho_\infty$.
For the first integral, we can estimate (because $\eta$ increases)
\begin{equation*}
2|\eta(\rho_\infty-\delta)|^{-\varepsilon} 
\geq 
\left|\int_{\rho_\infty-\delta}^{\rho_\infty} \frac{\mathrm d}{\mathrm d\rho} \eta^{-\varepsilon}\right|
= \varepsilon \int_{\rho_\infty-\delta}^{\rho_\infty} \left(2rU^2\eta^{-1-\varepsilon} + \zeta\eta^{-\varepsilon} + r\kappa(\zeta-N)\eta^{-1-\varepsilon}\right),
\end{equation*}
and since all the integrands on the right hand side are non-negative, their separate integrals must all be finite, in particular the one involving $U$.
For the second integral, we write
\begin{align*}
\eta = -\frac{r^{1-\varepsilon}}{\varepsilon}  \frac{\mathrm d}{\mathrm d\rho} r^\varepsilon - r\zeta 
= -r^{1-\varepsilon} \left( \frac{1}{\varepsilon}  \frac{\mathrm d}{\mathrm d\rho} r^\varepsilon + r^\varepsilon \zeta \right),
\end{align*}
so that the integral of $\eta$ near $\rho_\infty$ is finite, and H\"older's inequality implies
\begin{equation*}
\int_{\rho_\infty-\delta}^{\rho_\infty} r\eta^{1-\varepsilon}
\leq \left(  \int_{\rho_\infty-\delta}^{\rho_\infty} r^{\frac{1}{\varepsilon}} \right)^\varepsilon \left(  \int_{\rho_\infty-\delta}^{\rho_\infty} \eta \right)^{1-\varepsilon},
\end{equation*}
and the latter is finite since $r > 0$ is decreasing.
Thus, $w$ is bounded near $\rho_\infty$, as desired.
\end{proof}

In view of Lemma \ref{lemma-sing}, we may assume for the rest of the proof of Theorem \ref{thm-classification} that the solution remains in the regions $N + \zeta \geq 0$ and $|w|\leq1$.

\begin{lemma}\label{lemma-asymptotic-good}
Consider a solution of the initial value problem {\upshape (\ref{eq-r}--\ref{eq-zeta}, \ref{eq-initial-vals})} with respect to fixed initial data in $\mathscr I_0$.
Assume the solution satisfies $N + \zeta \geq 0$ and $|w|\leq1$ for all $\rho\geq 0$ for which it exists.
Then the solution is well-defined for all $\rho \geq 0$, and furthermore:

\begin{enumerate}
\item[(i)] the integral $\int_0^\infty N$ exists and has finite negative part,
\item[(ii)] $r$ has a non-zero limit at infinity, which is finite if and only if $N \to 0$,
\item[(iii)] all other dependent variables remain bounded as $\rho \to \infty$.
\end{enumerate}
\end{lemma}

\begin{remark*}
We would also like to point out that, a priori, $N$ could oscillate, i.e.\ the integral could be of the form $\infty-\infty$, so the existence in (i) is a non-trivial matter.
In part (iii), we make no claims about the existence of limits at infinity - this will be studied in the subsequent lemmata.
\end{remark*}

\begin{proof}
By the bounds in Lemma \ref{lemma-ineq}, we see that $|N|$ stays bounded as long as $N+\zeta \geq 0$, and therefore the solution is well-defined for all $\rho\geq 0$ by Lemma \ref{lemma-N-bounded-finite}.

By the monotone convergence theorem and Lemma \ref{lemma-ineq}, we have
\begin{equation*}
\lim_{\rho\to\infty} \int_0^\rho \zeta 
= 
\int_0^\infty \zeta 
\leq
\sqrt{E_0} \int_0^\infty \sech(\rho) = \frac{\pi\sqrt{E_0}}{2} < \infty.
\end{equation*}
On the other hand, since $N+\zeta\geq0$ by assumption, we also see by the monotone convergence theorem that
\begin{equation*}
\lim_{\rho\to\infty} \int_0^\rho (N+\zeta)
=
\int_0^\infty (N+\zeta)
\end{equation*}
where the integral on the right-hand side could be infinite.
Thus, 
\begin{align*}
\lim_{\rho\to\infty}\int_0^\rho N =
\lim_{\rho\to\infty}\left[\int_0^\rho (N+\zeta) - \int_0^\rho \zeta \right] =
\int_0^\infty (N+\zeta)  - \int_0^\infty \zeta = \int_0^\infty N,
\end{align*}
where we may take the limit on each term separately since they both have definite sign, and the negative part, i.e.\ the integral of $\zeta$, has finite limit.
Note that this also implies that the negative part of the integral of $N$ is finite and in particular the (Lebesgue) integral of $N$ over $[0,\infty)$ exist (but could be infinite).

Now equation (\ref{eq-r}) implies
\begin{equation*}
r(\rho) = r_0 \exp \int_0^\rho N  \to r_0 \exp \int_0^\infty N ,
\end{equation*}
so that $r$ has a limit at infinity, which is non-zero because the integral of $N$ cannot be negatively infinite.
For the second claim in (ii), note that if $N \to 0$, then the constraint (\ref{eq-kappa-constraint}) 
shows that
\begin{equation*}
\liminf_{\rho\to\infty} \frac{(1-w^2)^2}{r^2} \geq 1+\liminf_{\rho\to\infty} 2U^2 \geq 1,
\end{equation*}
since also $\zeta \to 0$ and $\kappa$ is bounded due to the inequality $\kappa + N \leq 2 + \sqrt{E_0}$.
This implies that $r$ cannot be unbounded (note that $|w|\leq1$ for $N+\zeta \geq 0$ by Lemma \ref{lemma-N-bounded-finite}) and hence has a finite limit by the preceding part of the lemma.
On the other hand, if $r$ has a finite limit, then so does $\log \frac{r}{r_0} = \int_0^\rho N$,
and since $\dot N$ is bounded by (\ref{eq-N}), it follows that $N\to 0$ by Barbălat's lemma (Lemma \ref{lemma-barbalat}).

For (iii), we first note that that $\kappa \leq 2 + \sqrt{E_0} - N$ is bounded since $-\zeta \leq N \leq 1$, where the first inequality follows by assumption and the second inequality follows from Lemma \ref{lemma-ineq}. 
It follows that the energy (\ref{eq-energy}) is bounded.
Furthermore, $|w| \leq 1$ because the assumption $N+\zeta \geq 0$ implies that $w$ cannot exit this region, cf.\ Lemma \ref{lemma-sing} (i).
Since $r$ stays away from zero at infinity by the already proven part of the lemma, we see that $\tfrac{(1-w^2)^2}{r^2}$ is bounded, and thus so is $2U^2 = \tfrac{(1-w^2)^2}{r^2}-1+E$.
\end{proof}

Now we wish to consider three cases separately:
\begin{enumerate}
\item $N$ has infinitely many zeros,
\item $N$ has finitely many zeros and $N<0$ after the last zero,
\item $N$ has finitely many zeros and $N>0$ after the last zero. 
\end{enumerate}
Before dealing with each of these cases, we would like to prove some preparatory results.
First, we provide a lemma which would thematically be a better fit for the next section, but we will use it in this section as well, so we state and prove it somewhat ahead of time.

\begin{lemma}\label{lemma-N-positive-r0-geq1}
Consider a solution of the initial value problem {\upshape (\ref{eq-r}--\ref{eq-zeta}, \ref{eq-initial-vals})} with respect to fixed initial data in $\mathscr I_0$.
If $r_0 \geq 1$, then either
\begin{itemize}
    \item $r\equiv 1$ and $w\equiv 0$, or
    \item $N>0$ for all $\rho > 0$ such that $|w|\leq 1$.
\end{itemize}
\end{lemma}

\begin{proof}
If $(r_0,w_0,U_0) = (1,0,0)$, then the solution is trivial with $r\equiv 1$ and $w\equiv 0$, cf.\ \S \ref{subsec-trivial-solutions}. 
Assume therefore that $r_0 \geq 1$ and $(r_0,w_0,U_0) \not= (1,0,0)$. 
By (\ref{eq-N}), we have
\begin{equation*}
\dot N(0) = 1-\frac{(1-w_0^2)^2}{r_0^2} \geq 1 - (1-w_0^2)^2 \geq 0.
\end{equation*}
If $r_0 > 1$, then the first inequality is strict and $N>0$ for small $\rho > 0$ since $N(0)=0$ by our choice of initial conditions.
If $r_0 = 1$ and $|w_0|\not=0$, then the second inequality is strict and again $N > 0$ for small $\rho > 0$.
If $r_0 = 1$ and $w_0=0$, then we can continue differentiating (\ref{eq-N}) to see that that $\dot N(0) = \ddot N(0) = 0$ but $\dddot N(0) = 4U_0^2$, and the latter is positive since otherwise $(r_0,w_0,U_0) = (1,0,0)$.
Thus, $N > 0$ for small $\rho > 0$ in all cases.

Suppose that $N$ ever reaches zero again in the region $|w|\leq 1$, so that there exists $\bar\rho > 0$ with $N > 0$ and $|w|\leq 1$ for $0 < \rho < \bar \rho$, as well as $N(\bar\rho) = 0$.
Then $r$ increases on this range and $r(\bar\rho) > r_0 \geq 1$, giving, by the same estimate as above, that $\dot N(\bar\rho) > 0$. 
This is a contradiction, so that we must have $N > 0$ for all $\rho >0$.
\end{proof}

The next result tells us that the condition $N \to 0$ as $\rho \to \infty$ is in fact a characterizing property of asymptotically cylindrical orbits.

%%%%%%%% N TO ZERO %%%%%%%%%

\begin{lemma}\label{lemma-N-tends-to-0}
Consider a solution of the initial value problem {\upshape (\ref{eq-r}--\ref{eq-zeta}, \ref{eq-initial-vals})} with respect to fixed initial data in $\mathscr I_0$.
If $N \to 0$ as $\rho \to \infty$, then the solution is asymptotically cylindrical, cf.\ Theorem \ref{thm-classification} {\normalfont (ii)}.
\end{lemma}

\begin{proof}
We first note that $r$ has a finite limit $0 < r_\infty \leq 1$ by Lemma \ref{lemma-asymptotic-good} (ii).
Furthermore, the energy (\ref{eq-energy}) tends to $0$ (since $\kappa$ is bounded and $\zeta \to 0$).
Now the autonomous energy (\ref{eq-aut-energy}) satisfies $F = r^2(E-1) \to -r_\infty^2$.
For any $\rho_0 \geq 0$, we also have
\begin{equation*}
F(\rho_0) + F(\rho) = -\int_{\rho_0}^\rho \dot F = 4\int_{\rho_0}^\rho (\kappa-2N)\dot w^2 .
\end{equation*}
If $\rho_0$ is selected so large that $\kappa - 2N \geq c$ for $\rho \geq \rho_0$ and some constant $c >0$ (this is possible since $\liminf \kappa \geq 1$ and $N\to 0$), then we see by the monotone convergence theorem that
\begin{equation*}
\int_{\rho_0}^\infty \dot w^2 = \lim_{\rho\to\infty} \int_{\rho_0}^\rho \dot w^2
\leq \frac{1}{4c}\lim_{\rho\to\infty} (F(\rho_0)+F(\rho)) = \frac{F(\rho_0)+r_\infty^2}{4c},
\end{equation*}
and hence $\dot w$ is square-integrable near infinity.
Since $$\ddot w = -(\kappa-2N)\dot w - w(1-w^2)$$ is bounded (note that $\dot w$ is bounded because $F$ and $w$ are), we see that $\dot w^2$ is uniformly continuous, and hence $\dot w \to 0$ by Barbălat's lemma (Lemma \ref{lemma-barbalat}). 
Thus $U\to 0$, as well as $\kappa \to 1$ by Lemma \ref{lemma-kappa-to-1}.

Now $(1-w^2)^2 = 2\dot w^2 - F$ has a limit at infinity, and hence $w$ also tends to some limit $|w_\infty| \leq 1$ by continuity. 
By equation (\ref{eq-w}), we must have $w_\infty \in \{0, \pm 1\}$ since $\ddot w \to -w_\infty(1-w_\infty^2)$, and any other choice of $w_\infty$ would contradict $\dot w \to 0$.
But $w_\infty = \pm 1$ is impossible since that would imply the absurdity $F \to 0 = -r_\infty^2$. 
It follows that $w \to 0$, and from $E\to 0$ we also get $r \to 1$.

Next, we observe that the only asymptotically cylindrical solution with $r_0 \geq 1$ is the trivial orbit with $r\equiv 1$ and $w\equiv 0$ by Lemma \ref{lemma-N-positive-r0-geq1} (since otherwise $N>0$ and $r>1$ for $\rho > 0$ with $|w|\leq 1$, and we cannot have $r\to 1$).
On the other hand, for $r_0<1$ the trivial solution $w\equiv 0$ has $E_0 < 0$ and thus is not admissible.
Therefore it only remains to show that $w$ truly has infinitely many zeros in the case $r_0 < 1$.
To this end, we consider the polar angle defined by
\begin{equation*}
-2\pi < \theta(0) \leq 0, \qquad \tan\theta = \frac{\dot w}{w} \quad\text{if}\quad w\not=0,
\end{equation*}
and extended smoothly across zeros of $w$.
Note that this is well-defined because $(w,\dot w)$ stays away from the origin, since $w\not\equiv 0$.
A simple calculation yields
\begin{equation}\label{eq-theta}
\dot\theta + 1 = -(\kappa - 2N)\frac{w\dot w}{w^2+\dot w^2} + \frac{w^4}{w^2+\dot w^2}.
\end{equation}
Now since $2|w\dot w| \leq w^2+\dot w^2$,
\begin{equation*}
\frac{|w\dot w|}{w^2+\dot w^2} \leq \frac{1}{2} \quad\text{and}\quad
\frac{w^4}{w^2+\dot w^2} = \frac{w^2}{1 + (\dot w/w)^2} \leq w^2
\end{equation*}
we get
\begin{equation*}
\dot\theta + \frac12 \leq \frac12|\kappa-2N-1| + w^2.
\end{equation*}
But the right hand side of the latter inequality tends to $0$, so $\limsup\dot\theta\leq -\frac12$, implying that $\theta\to-\infty$, and thus $w$ necessarily crosses zero infinitely many times.
\end{proof}

%%%%%%%% INFINITELY MANY ZEROES %%%%%%%%%%%

With Lemma \ref{lemma-N-tends-to-0} in hand, we can now continue the proof of Theorem \ref{thm-classification}, by tackling the different cases depending on the number of zeros of $N$, as already hinted.

\begin{lemma}\label{lemma-N-infinite-zeros}
Consider a solution of the initial value problem {\upshape (\ref{eq-r}--\ref{eq-zeta}, \ref{eq-initial-vals})} with respect to fixed initial data in $\mathscr I_0$.
If $N$ has infinitely many zeros, then the solution is asymptotically cylindrical, cf.\ Theorem \ref{thm-classification} {\normalfont (ii)}.
\end{lemma}

\begin{proof}
Let $\rho_n$ be the increasing sequence of zeros of $N$.
Note that $\rho_n$ must be unbounded by Lemma \ref{lemma-N-bounded-finite} and the fact that the region $N < -\zeta$ is invariant, so that $N$ can only have zeros while staying in the region $-\zeta \leq N \leq 1$.
Since $r$ tends to some limit at infinity by Lemma \ref{lemma-asymptotic-good} (ii), we may use the constraint (\ref{eq-kappa-constraint}) to calculate
\begin{equation*}
\lim_{\rho\to\infty} r(\rho)^2 = \lim_{n\to\infty} r(\rho_n)^2
= \lim_{n\to\infty}\frac{(1-w(\rho_n)^2)^2}{1+2U(\rho_n)^2 - \zeta(\rho_n)^2}
\leq \limsup_{n\to\infty} \,(1-w(\rho_n)^2)^2 \leq 1
\end{equation*}
which is finite since $w$ is bounded by Lemma \ref{lemma-asymptotic-good} (iii).
This implies that the limit of $r$ is finite, and thus $N \to 0$ by Lemma \ref{lemma-asymptotic-good} (ii).
Hence, we may apply Lemma \ref{lemma-N-tends-to-0} to conclude.
\end{proof}

%%%%%%% FINITELY MANY ZEROES (-) %%%%%%%%%%%%%%

\begin{lemma}\label{lemma-N-neg}
Consider a solution of the initial value problem {\upshape (\ref{eq-r}--\ref{eq-zeta}, \ref{eq-initial-vals})} with respect to fixed initial data in $\mathscr I_0$.
Assume there is a point $\rho_0 \geq 0$ such that $N(\rho_0) = 0$ and $N(\rho) < 0$ for $\rho > \rho_0$. Then the solution is singular, cf.\ Theorem \ref{thm-classification} {\normalfont (i)}.
\end{lemma}

\begin{proof}
In view of Lemma \ref{lemma-sing}, we only need to show that $N$ exits the region $-\zeta \leq N < 0$. 
Aiming to reach a contradiction, assume that $N$ remains in this region for all $\rho > \rho_0$. 
Then we also have $|w| \leq 1$ in view of Lemma \ref{lemma-sing} (i).
Furthermore, we see that $N \to 0$ as $\rho\to\infty$ (since $\zeta \to 0)$, and hence $r\to 1$ by Lemma \ref{lemma-N-tends-to-0}. 
But $N$ is negative for $\rho > \rho_0$, so we must have $r > 1$ on this range, in particular also at $\rho=\rho_0$ where $N(\rho_0)=0$ and so
\begin{equation*}
\dot N(\rho_0) = 1 - \frac{(1-w(\rho_0)^2)^2}{r(\rho_0)^2} > 0,
\end{equation*}
implying that $N>0$ for sufficiently close $\rho \geq \rho_0$, which is a contradiction.  
\end{proof}

%%%%%%%% FINITELY MANY ZEROES (+) %%%%%%%%%%%%%
\begin{lemma}\label{lemma-N>0}
Consider a solution of the initial value problem {\upshape (\ref{eq-r}--\ref{eq-zeta}, \ref{eq-initial-vals})} with respect to fixed initial data in $\mathscr I_0$.
Assume there is a point $\rho_0 \geq 0$ such that $N(\rho) > 0$ for all $\rho > \rho_0$.
\begin{enumerate}
\item[(i)] If $r$ is bounded, then the solution is asymptotically cylindrical, cf.\ \mbox{Theorem \ref{thm-classification} {\normalfont (ii)}.}
\item[(ii)] If $r$ is unbounded, then the solution is asymptotically flat, cf.\ Theorem \ref{thm-classification} {\normalfont (iii)}.
\end{enumerate}
\end{lemma}

\begin{proof}
Part (i) follows trivially from Lemma \ref{lemma-N-tends-to-0} since $N \to 0$ by Lemma \ref{lemma-asymptotic-good} (ii).
For part (ii), we first show that $U \to 0$.
We consider separately the cases where $w$ has finitely or infinitely many zeros.

If $w$ has finitely many zeros, then there is a $\bar\rho \geq 0$ such that, 
without loss of generality, $0 < w(\rho) \leq 1$ for $\rho \geq \bar\rho$.
Define the function
\begin{equation*}
v(\rho) = \int_{\bar\rho}^\rho U = \frac{w(\rho)}{r(\rho)} - \frac{w(\bar\rho)}{r(\bar\rho)} + \int_{\bar\rho}^\rho \frac{w\dot r}{r^2}, 
\end{equation*}
where we integrate by parts in the last equality.
Note that $w$ is bounded and $r \to \infty$, so that the first term on the right-hand side tends to 0, while the final integral also has a finite limit since it is a bounded increasing function of $\rho$.
Thus $v$ has a finite limit as infinity, and since $\ddot v = \dot U$ is bounded by Lemma \ref{lemma-asymptotic-good} (iii) and (\ref{eq-U}), it follows from Barbălat's lemma \cite{barbalat} that $\dot v = U \to 0$ in this case.

If $w$ has infinitely many zeros, then so does $U$ and we can find a sequence $\rho_k \to \infty$ with $U(\rho_k) = 0$.
Since $\kappa-N \geq 0$, we see that the energy satisfies
\begin{equation*}
\dot E = -4(\kappa-N)U^2 + \frac{N(1-w^2)^2}{r^2}  \leq \frac{N(1-w^2)^2}{r^2},
\end{equation*}
and hence
\begin{equation*}
E(\rho) - E(\rho_k) \leq \int_{\rho_k}^\rho \frac{N(1-w^2)^2}{r^2} \leq \int_{\rho_k}^\rho \frac{\dot r}{r^3} = \frac{2}{r(\rho_k)^2} - \frac{2}{r(\rho)^2},
\end{equation*}
since $|w|\leq 1$ and $N>0$ for $\rho \geq \rho_k$ if $k$ is large enough.
Thus, since $U(\rho_k) = 0$,
\begin{equation*}
\limsup_{\rho\to\infty} E(\rho) \leq  E(\rho_k)+\frac{2}{r(\rho_k)^2} 
= 1 - \frac{(1-w(\rho_k)^2)^2}{r(\rho_k)^2} +\frac{2}{r(\rho_k)^2} \to 1,
\end{equation*}
where in the end we let $k \to \infty$.
Since trivially $\liminf E \geq 1$ (because $r$ is unbounded and $w$ is bounded), we see that $E \to 1$. 
This implies that
$2U^2 = E - 1 + \frac{(1-w^2)^2}{r^2} \to 0$
in this case as well.

By Lemma \ref{lemma-kappa-to-1}, we now get $\kappa \to 1$.
From the constraint (\ref{eq-kappa-constraint}), we see that
\begin{equation*}
(1-N)^2 = 1 - E + 2N(\kappa-1) + \zeta^2 \to 0,
\end{equation*}
and hence $N \to 1$. 
Next, we wish to show that $\dot w \to 0$. 
Note that this does not follow directly from the fact that $U = \dot w/r \to 0$ since $r$ is unbounded.
It is not even clear, prima facie, whether $\dot w$ is bounded - this does not follow from Lemma \ref{lemma-asymptotic-good} (iii), since we do not consider $\dot w$ as one of the dependent variables.

Since $\kappa-2N \to -1$ by the proof above, we may choose $\bar\rho \geq \rho_0$ so that $\kappa - 2N \leq -\tfrac12$ for $\rho\geq \bar\rho$.
The autonomous energy $F=2\dot w^2-(1-w^2)^2$ satisfies
\begin{equation*}
\dot F = -4(\kappa-2N)\dot w^2 \geq 2\dot w^2 = F+(1-w^2)^2 \geq F.
\end{equation*}
This implies, in particular, that the region $F \geq \varepsilon$ is invariant for any $\varepsilon > 0$.
In this region, we have
$2\dot w^2 = F + (1-w^2)^2 \geq \varepsilon>0$.
Thus, if the orbit enters the region $F\geq \varepsilon$, then $w$ exits the strip $|w|\leq 1$ at some finite $\rho$, which is a contradiction.
Consequently, $F \leq 0$ for all $\rho \geq \bar\rho$, which also implies that $\dot w$ is bounded.%
\footnote{This argument is similar in spirit to \cite[Lemma 3--4]{bizon-popp}.}
Now
\begin{equation*}
\int_{\bar\rho}^\rho \dot w^2 \leq -2\int_{\bar\rho}^\rho (\kappa-2N)\dot w^2 = \frac12 \int_{\bar\rho}^\rho \dot F \leq -\frac12 F(\bar\rho)
\end{equation*}
and applying the monotone convergence theorem with $\rho \to \infty$ shows that $\dot w$ is square-integrable.
Since $\ddot w$ is bounded, $\dot w^2$ is also uniformly continuous, and we get $\dot w \to 0$ by Barbălat's lemma (Lemma \ref{lemma-barbalat}).

Now since $F$ is non-decreasing for $\rho \geq \bar\rho$ and $F \leq 0$, we see that $F$ has a limit at infinity.
Hence, $(1-w^2)^2 = 2\dot w^2 - F$ also has a limit at infinity, and consequently $w$ also tends to some limit $|w_\infty| \leq 1$ by continuity.
By (\ref{eq-w}) we see that $\ddot w \to -w_\infty(1-w_\infty^2)$ and consequently $w_\infty \in \{0, \pm 1\}$ (otherwise we could not have $\dot w \to 0$).

If $w_\infty = 0$, we need to show that the solution is the trivial $w\equiv 0$.
In fact, we see that $F \to -1$, and since $F$ is non-decreasing for large $\rho$ and $F \geq -1$, we necessarily have $F\equiv -1$ which implies $w\equiv 0$.
The discussion in \S \ref{subsec-trivial-solutions} also implies that in this case $r_0 > 1$.

Finally, we need to study the limits of $r\zeta$ and $r(1-N^2)$ at infinity.
Note that
\begin{equation*}
\frac{\mathrm d}{\mathrm d\rho} \log (r\zeta) = N-\kappa \leq 0
\end{equation*}
by Lemma \ref{lemma-ineq}, so $r\zeta$ is non-increasing and
\begin{equation*}
\lim_{\rho\to\infty} r\zeta = r_0\sqrt{E_0} \, \exp \int_0^\infty (N-\kappa)
\end{equation*}
by the monotone convergence theorem.
Therefore, to show that $r\zeta$ has a finite positive limit, it suffices to show that the integral on the right-hand side converges.
By (\ref{eq-N-2}),
\begin{equation*}
N-\kappa = -\frac{1}{N} (\dot N - 2U^2 + \zeta^2).
\end{equation*}
Hence, the claim will follow if we can show that of $U$ and $\zeta$ are square integrable near infinity, since $N\to1$.
But the integrability of $U^2$ follows trivially from the integrability of $\dot w^2 = (rU)^2$ (since $r\to\infty$), whereas the integrability of $\zeta^2$ follows directly from the bound $\zeta^2 \leq E_0 \sech^2(\rho)$, cf.\ Lemma \ref{lemma-ineq}.
On the other hand, we have
\begin{equation}\label{eq-mass-limit}
r(1-N^2) = 2m = 2\int_0^\rho \dot m = r_0 + \int_0^\rho \frac{\dot r}{r^2}\left[2\dot w^2 + (1-w^2)^2 - (r\zeta)^2\right],
\end{equation}
and the integral converges as $\rho \to 0$ since the quantity in the parentheses converges the finite quantity $(1-w_\infty^2)^2 - \alpha^2$.
Therefore $r(1-N^2)$ has a finite limit at infinity by dominated convergence, and this limit must be non-negative since $N \to 1$ and $N \leq 1$ (cf.\ Lemma \ref{lemma-ineq}).
\end{proof}

We now have essentially all the ingredients necessary to complete the proof of Theorem \ref{thm-classification}.
To summarize, the pieces of the puzzle can be put together as follows.
In Lemma \ref{lemma-N-bounded-finite} we showed that the solution can be continued as long as $N$ remains finite, and in particular the only way that a solution stops existing at a finite point $\rho=\rho_\infty$ is if $N \to \infty$ (and $r\to 0$), which is precisely the behaviour describing the singular orbits from Theorem \ref{thm-classification} (i).
In Lemma \ref{lemma-sing} we showed that such singular behaviour happens precisely when the solution enters the region $N+\zeta < 0$, and we also showed that the solution enters this region if it exits the strip $|w|\leq 1$.
This allowed us to assume that the solution satisfies $N+\zeta \geq 0$ and $|w|\leq 1$, in which case it must be globally defined by Lemma \ref{lemma-N-bounded-finite} and the bounds from Lemma \ref{lemma-ineq}.
In Lemma \ref{lemma-asymptotic-good} we showed that the remaining dependent variables remain globally bounded and, most importantly, that $r$ has a limit (which may be infinite) at infinity.
The remainder of the proof then revolved around studying the different possible behaviours of $N$, depending on its number of zeros: 
\begin{itemize}
\item In Lemma \ref{lemma-N-infinite-zeros} we treat the case when $N$ has infinitely many zeros, and show that the solution is then asymptotically cylindrical as in Theorem \ref{thm-classification} (ii);
\item In Lemma \ref{lemma-N-neg} we treat the case when $N$ has finitely many zeros and is negative after the last zero, in which case the solution is singular as in Theorem \ref{thm-classification} (i);
\item In Lemma \ref{lemma-N>0} we treat the final case when $N$ has finitely many zeros and is positive after the last zero, in which case the solution asymptotically cylindrical Theorem \ref{thm-classification} (ii) if $r$ is bounded, and asymptotically flat as in Theorem \ref{thm-classification} (iii) if $r$ is unbounded.
\end{itemize}
Since this exhausts all the possible cases, the proof of Theorem \ref{thm-classification} is complete.

\section{Neighbourhood theorem}
\label{sec-neighbourhood-theorem}
% !TEX root = main.tex

In this section, we will study the set of admissible initial data
\begin{equation*}
\mathscr I_0 = \left\{ (r_0, w_0, U_0) \in \mathbb R^3 \;\mid\; r_0 > 0, \; |w_0| \leq 1 , \; E_0 \geq 0\right\},
\end{equation*}
and particularly its subsets generated by the classification in Theorem \ref{thm-classification}.
We will hereafter fix $r_0>0$, and a given subset $\mathscr Y \subset \mathscr I_0$ we will write
\begin{equation*}
\mathscr Y(r_0) = \{ (w_0,U_0) \in \mathbb R^2 \,\mid\, (r_0,w_0,U_0) \in \mathscr Y \}.
\end{equation*}
We also introduce the following notation, which slightly differs from the categorization in Theorem \ref{thm-classification}, but it will turn out to be more convenient throughout this section.

\begin{definition}\label{def-solution-classes}
We define the following subsets of $\mathscr I_0$:
\begin{itemize}
\item For the \textit{singular orbits}, i.e.\ those with $r\to 0$ and $N\to-\infty$ at some finite $\rho$, we define the subsets (denoting by $n$ the number of zeros of $w$ for $\rho > 0$):
\vskip0.1cm
\begin{itemize}
\item $\mathscr E_n$ as the set of \textit{escaping singular orbits}, for which $w$ escapes the strip $|w|\leq 1$.
\item $\mathscr C_n$ as the set of \textit{crashing singular orbits}, for which $w$ stays in $|w|\leq 1$.
We also put $\mathscr C = \bigcup_{n=0}^\infty \mathscr C_n$.
\end{itemize}
\item $\mathscr O$ as the set of \textit{oscillatory orbits}, defined for all $\rho>0$ with $(w,\dot w) \to (0,0)$.
\item $\mathscr R_n$ as the set of \textit{regular orbits}, defined for all $\rho > 0$ with $(|w|,\dot w) \to (1,0)$, and such that $w$ has $n$ zeros for $\rho > 0$.
\end{itemize}
\end{definition}

By Theorem \ref{thm-classification}, we see that the sets $\mathscr R_n, \mathscr E_n, \mathscr C, \mathscr O$ form a disjoint partition of $\mathscr I_0$.
We also observe that all orbits in $\mathscr R_n$ are asymptotically flat.
For orbits in $\mathscr O(r_0)$, we have that:
\begin{itemize}
\item if $r_0 > 1$, then $w\equiv 0$ and the orbit is asymptotically flat, 
\item if $r_0 = 1$, then $w\equiv 0$ and the orbit is asymptotically cylindrical,
\item if $r_0 < 1$, then $w \not\equiv 0$ and the orbit is asymptotically cylindrical. 
\end{itemize}
Finally, we note that $\mathscr C(r_0) = \varnothing$ for $r_0 \geq 1$, cf.\ Lemma \ref{lemma-N-positive-r0-geq1}.

Our current goal is to show the following result, which tells us how the orbits near a given orbit type behave.

\begin{theorem}[Neighbourhood theorem]\label{thm-ngbh-initial-sets}
Fix $r_0 > 0$ and let $x_0 = (w_0,U_0) \in \mathscr I_0(r_0)$.
Define the slices
\begin{align*}
\mathscr W_\delta(r_0, x_0) &= \{ (\lambda, U_0) \in \mathscr I_0(r_0) \,\mid\, |\lambda-w_0|\leq \delta \},\\[0.1cm]
\mathscr U_\delta(r_0, x_0) &= \{ (w_0, \mu) \in \mathscr I_0(r_0) \,\mid\, |\mu-U_0|\leq \delta \}.
\end{align*}
\begin{enumerate}
\item Let either
\begin{equation*}
\mathscr Y(r_0) = \mathscr E_n(r_0) \quad \text{or}\quad \mathscr Y = \mathscr C(r_0) \setminus \{0\}.
\end{equation*}
If $x_0 \in \mathscr Y(r_0)$ and $\delta>0$ is sufficiently small, we have
\begin{equation*}
\mathscr U_\delta(r_0, x_0) \subset \mathscr Y(r_0).
\end{equation*}
If we also have $w_0\not=0$, then
\begin{equation*}
\mathscr W_\delta(r_0, x_0) \subset \mathscr Y(r_0).
\end{equation*}
\item If $x_0 \in \mathscr O(r_0)$, then for any $n \geq 0$ we can choose $\delta>0$ so small that
the $y_0$-orbit has at least $n$ zeros of $w$ for $\rho>0$,
\iffalse
\begin{equation*}
y_0 \in \begin{cases}
\bigcup_{m=n}^\infty \mathscr R_m \cup \mathscr E_m, & \text{if } r_0 \geq 1,\\[0.1cm]
\bigcup_{m=n}^\infty \mathscr R_m \cup \mathscr E_m \cup \mathscr C_m \cup \mathscr O, & \text{if } r_0 < 1.
\end{cases},
\end{equation*}
\fi
whenever $0 < |y_0 - x_0| \leq \delta$.
\item If $x_0 \in \mathscr R_n(r_0) \setminus \{0\}$ and $\delta>0$ is sufficiently small, we have
\begin{equation*}
\mathscr U_\delta(r_0, x_0) \subset  (\mathscr R_n \cup \mathscr E_n \cup \mathscr E_{n+1})(r_0).
\end{equation*}
If we also have $w_0\not=0$, then
\begin{equation*}
\mathscr W_\delta(r_0, x_0) \subset  (\mathscr R_n \cup \mathscr E_n \cup \mathscr E_{n+1})(r_0).
\end{equation*}
\end{enumerate}
\end{theorem}

\begin{remark}\label{remark-uniqueness-regular-orbits}
This result is an analogue of \cite[Propositions 30 and 32]{breit-forg-mais} for the wormhole setting.
Part (i) essentially says that $\mathscr C(r_0)$ and each $\mathscr E_n(r_0)$ are open relative to the lines where either $w_0$ is fixed, or $U_0$ is fixed and $w_0\not=0$.
In parts (ii) and (iii), one can actually get a stronger result.
In fact, one can show that oscillatory and regular orbits are locally unique (cf.\ \cite[Propositions 31 and 33]{breit-forg-mais}).
This fact is, however, non-essential for the proofs in \S \ref{sec-construction}, nor does it improve the statements of the corresponding theorems, so we omit them.%
\footnote{The authors of the cited paper likely included this analysis as it seems like they were also attempting to get some global uniqueness result for regular orbits with a given number of zeros of $w$.}
Finally, we would like to point out that the proof of part (iii) given in \cite[Lemma 20]{breit-forg-mais} for their setting again invokes the theory of structurally stable vector fields, which we have already commented on in Remark \ref{remark-bfm-gap}.
The proof given here naturally avoids this.
\end{remark}

\begin{remark}\label{remark-close-orbits}
Throughout the proofs in this section, we will repeatedly use some basic facts about ordinary differential equations.
The solution flow depends continuously on the initial data $(w_0,U_0) \in \mathscr I_0(r_0)$, as we have already observed in Remark \ref{remark-continuity-initial-data}.
Furthermore, classical results \cite[Theorem 3.2]{hartman} ensure that the maximal forward point of existence $\bar \rho > 0$ of the solution is a lower-semicontinuous function of the initial data.
This is important to us because we would like to compare the values of an orbit with its nearby orbits near the end of its existence.

Consider now an orbit with initial data $x_0=(w_0, U_0) \in \mathscr I_0$.
The $x_0$-orbit is either singular and hence defined up to some finite $\rho_\infty$, or it exists for all $\rho \geq 0$.
As already noted, we can choose $\delta > 0$ so that the solutions with initial data $y_0$ such that $|y_0 - x_0| \leq \delta$ are defined up to any $\bar\rho < \rho_\infty$ in the former case, or up to $\bar\rho$ as large as we would like in the latter case.  

By further shrinking $\delta > 0$ if necessary, we can also ensure that the values of the dependent variables at $\bar\rho$ differ by no more than any given $\varepsilon > 0$, since the solutions depend continuously on the initial parameters.
Moreover, if the $x_0$-orbit is defined for all $\rho \geq 0$ and a continuous function of its dependent variables, call it $\xi$, tends to some limit $L$, then we can for any $\varepsilon, K > 0$ find $\bar\rho>0$ (depending only on $\varepsilon$) so large and $\delta>0$ so small that $\xi(\rho,y_0)$ differs from $L$ by no more than $\varepsilon$ on the interval $\bar\rho \leq \rho \leq \bar\rho + K$.
Indeed, we can first choose $\bar\rho$ so large that $|\xi(\rho,x_0)-L|\leq \varepsilon/2$ for $\rho\geq\bar\rho$, and $\delta$ small enough that orbits with initial data in $\{|y_0-x_0|\leq \delta\}$ are defined at least up to $\bar\rho+K$.
Then $\xi$ is continuous on the compact set $$[\bar\rho, \bar\rho+K] \times \{|y_0-x_0|\leq \delta\},$$ so it is also uniformly continuous there.
By shrinking $\delta$ further if necessary we then get $|\xi(\rho,y_0)-\xi(\rho,x_0)| \leq \varepsilon/2$ for all $(\rho,y_0) \in [\bar\rho, \bar\rho+K] \times \{|y_0-x_0|\leq \delta\}$, which gives the desired claim.

Finally, we note that if $w(\rho,x_0)$ has (exactly) $n$ zeros for $0 < \rho < \bar\rho$  and $w(\bar\rho,x_0)\not=0$ (in particular $w\not\equiv 0$), then for sufficiently small $\delta$ and initial data $y_0$ with either $y_0 \in \mathscr U_\delta(r_0,x_0)$, or $y_0 \in \mathscr W_\delta(r_0,x_0)$ with $w_0\not=0$,
the corresponding solution $w(\rho,y_0)$ has (exactly) $n$ zeros for $0 < \rho < \bar\rho$.
To see this, let $\{\rho_k \mid k=1,\ldots,n\}$ be the $n$ zeros of $w(\rho,x_0)$ on $0 < \rho < \bar\rho$.
If $w_0 = 0$ we also put $\rho_0 = 0$ and let $k$ run from $0$ to $n$ instead.
Then $\dot w(\rho_k, x_0) \not= 0$ since otherwise we would have $w\equiv 0$ (cf.\ Remark \ref{remark-symmetries}),
and we can by continuity find $\delta,\gamma > 0$ so small that 
\begin{equation*}
\dot w(\rho,y_0) \not= 0, \qquad \rho \in (\rho_k - \gamma, \rho_k + \gamma) =: I_k,
\end{equation*}
for each $k$ and all $y_0$-orbits (note that $\gamma$ can be chosen uniformly because $k$ runs over a finite set).
In particular, $w(\rho,y_0)$ is strictly monotone in each $I_k$. 
On the other hand, $\rho\mapsto|w(\rho,x_0)|$ is positive on the compact set
\begin{equation*}
K = [0,\bar\rho] \setminus \bigcup_k I_k,
\end{equation*}
so it also has a positive minimum on $K$ and we may select $\varepsilon>0$ strictly smaller than this minimum.
Then, finally, we shrink $\delta$ further if necessary to make
\begin{equation*}
|w(\rho,y_0) - w(\rho,x_0)| < \varepsilon,
\end{equation*}
which also ensures that, for each fixed $y_0$, the corresponding $\rho\mapsto w(\rho,y_0)$ has no zeros on $K$, and has exactly one zero in each $I_k$ by monotonicity and the intermediate value theorem (note that the zeros are no longer necessarily located at $\rho_k$ for the $y_0$-orbits).
This gives the claim if $w_0\not=0$.
On the other hand, if $w_0=0$, then we note that the zero at $\rho_0=0$ stays located at $\rho=0$ for the nearby $y_0$-orbits, since we only consider $y_0\in \mathscr U_\delta(r_0,x_0)$ in this case, and consequently the number of zeros for $\rho>0$ remains unaffected.%
\footnote{This is important in view of the way we count zeros of $w$ (cf.\ Definition \ref{def-solution-classes}).
It is also the main reason for not allowing the case $w_0=0$ when considering the slices $\mathscr W_\delta$ in the $w_0$-direction. Indeed, in this case the initial zero of $w$ could move around too much, which would also destroy the counting of the zeros near $\rho=0$.}
\end{remark}

\begin{proof}[Proof of Theorem \ref{thm-ngbh-initial-sets}]
%% ESC OR CRASH
Assume first that $x_0\in \mathscr E_n(r_0) \setminus\{0\}$ and choose $\delta > 0$ so small and $\bar\rho$ so close to the singular point that $(N+\zeta)(\bar\rho, y_0) < 0$ and $w(\bar\rho, y_0) > 1$, and $w$ has exactly $n$ zeros for the desired initial data $y_0 \in \mathscr U_\delta(r_0,x_0)$ or $y_0 \in \mathscr W_\delta(r_0,x_0)$ with $w_0\not=0$ (cf.\ Remark \ref{remark-close-orbits}).
All of the $y_0$-orbits are then singular by Lemma \ref{lemma-sing}, and $w(\rho,y_0)$ cannot gain any more zeros for $\rho>\bar\rho$ because it stays in the region $|w|>1$. 
It follows that $y_0\in \mathscr E_n(r_0)$.

An analogous argument can be applied when $x_0 \in \mathscr C(r_0) \setminus \{0\}$ to conclude that the nearby $y_0$-orbits are singular (but in this case $w(\rho,y_0)$ could gain zeros for $\rho>\bar\rho$, which is why we refrain from counting them in the first place).
Furthermore, the energy (\ref{eq-energy}) is non-increasing for $N \leq -\zeta$ and satisfies
\begin{equation}\label{eq-w-cannot-cross-1}
E = 2\kappa N - N^2 + \zeta^2 \leq 0, \quad\text{i.e.}\quad \frac{(1-w^2)^2}{r^2} \geq 1 + 2U^2 \geq 1,
\end{equation}
so $w$ cannot cross the lines $|w|=1$ in this region and consequently $y_0 \in \mathscr C(r_0) \setminus \{0\}$.

%% OSC

Now consider the oscillatory case (ii) of Theorem \ref{thm-ngbh-initial-sets}.
If $r_0<1$, the $x_0$-orbit has infinitely many zeros of $w$ while $w\not\equiv 0$, so the result follows trivially by continuity with respect to initial data, cf.\ Remark \ref{remark-close-orbits} (alternatively, one can use a similar argument as in the case $r_0=1$ below, if one prefers).
For $r_0 \geq 1$, we first note that the orbits with initial data $y_0$ in the punctured ball $0 < |y_0 - x_0| \leq \delta$ cannot be in $\mathscr O(r_0)$ or $\mathscr C(r_0)$, as follows from Theorem \ref{thm-classification} and the fact that no orbit can crash for $r_0\geq1$, cf.\ Lemma \ref{lemma-N-positive-r0-geq1}.
To show that these have arbitrarily many zeros, we consider the polar angle of $(w,\dot w)$ defined by
\begin{equation*}
-2\pi < \theta(0) \leq 0, \qquad \tan\theta = \frac{\dot w}{w} \quad\text{if}\quad w\not=0,
\end{equation*}
and extended smoothly across zeros of $w$.
Note that this is well-defined since the orbits starting in the punctured ball $0 < |y_0-x_0| \leq \delta$ cannot reach the fixed point $(w,\dot w)=(0,0)$.
A simple calculation yields (cf.\ (\ref{eq-theta})) 
\begin{equation*}
\dot\theta + \frac12 \leq \frac12\lvert\kappa - 2N \pm 1\rvert + w^2.
\end{equation*}
The idea is to show that the right-hand side is sufficiently small on an arbitrarily large interval, which will imply arbitrarily many zeros of $w$.
Note that:
\begin{itemize}
\item If $r_0=1$, then necessarily $x_0 = (1,0,0)$ and $w\equiv 0$, $N\equiv 0$, $\kappa = \tanh(\rho)\to1$.
\item If $r_0 > 1$, then necessarily $x_0 = (r_0,0,0)$ and $w \equiv 0$, while $\kappa\to1$ and $N\to 1$ since these orbits are asymptotically flat, cf.\ \S \ref{subsec-trivial-solutions}.
\end{itemize} 
Hence, in both cases we can for any $K>0$ find $\bar\rho>0$ sufficiently large (independent of $K$) and $\delta>0$ so small that (cf.\ Remark \ref{remark-close-orbits})
\begin{equation*}
|(\kappa-2N)(\rho,y_0) \pm 1|\leq \frac{1}{4} \quad\text{and}\quad w(\rho,y_0)^2 \leq \frac{1}{8},
\end{equation*}
for $\bar\rho \leq \rho \leq \bar\rho+K$ and all $y_0$-orbits with $|y_0-x_0|\leq \delta$, where the first inequality holds with the minus sign for $r_0=1$, and with the plus sign for $r_0 > 1$.
It therefore follows that $\dot \theta(\rho,y_0) \leq -\frac14$ for $\bar\rho \leq \rho \leq \bar\rho+K$ and $|y_0-x_0| \leq \delta$, implying that
\begin{equation*}
\theta(\bar\rho + K, y_0) = \theta(\bar\rho, y_0) + \int_{\bar\rho}^{\bar\rho+K}\dot\theta(\cdot, y_0) \leq \theta(\bar\rho, y_0) - \frac{K}{4}
\end{equation*}
But the right-hand side can be made arbitrarily (negatively) large, which implies arbitrarily many zeros of $w$.

Finally, we consider part (iii), so let $x_0\in \mathscr R_n(r_0) \setminus \{0\}$.
We first want to make sure that, for nearby orbits, $N$ and $\kappa$ stay near 1.
Since $E_0$ is bounded on $|y_0-x_0|\leq \delta$, we see that for any given $\bar\varepsilon \leq 1$ and sufficiently large $\bar\rho$ depending on it, all orbits with initial data in $|y_0-x_0|\leq \delta$ have 
\begin{equation*}
\sqrt{E_0(y_0)}\sech(\rho) \leq 2\bar\varepsilon \quad\text{and}\quad (\kappa+N)(\rho,y_0)\leq 2(1+\bar\varepsilon)
\end{equation*}
for all $\rho \geq \bar\rho$, where the latter inequality follows by the former and Lemma \ref{lemma-ineq}.
Since the $x_0$-orbit is asymptotically flat, it has $r\to\infty$ and $N\to 1$, so we may further increase $\bar\rho$ and shrink $\delta$ if necessary to ensure that 
$$N(\bar\rho, y_0) \geq 1 - \bar\varepsilon, \quad r(\bar\rho, y_0) \geq \frac{1}{\bar\varepsilon}$$
for all orbits with initial data in $|y_0-x_0|\leq \delta$.
As long as $N\geq0$ and $|w|\leq 1$, we then have $r \geq r(\bar\rho, y_0) \geq \frac{1}{\bar\varepsilon}$, and equation (\ref{eq-N}) gives
\begin{equation*}
\dot N \geq 1 - \bar\varepsilon^2 - [2(1+\bar\varepsilon)-N]N  \geq (1+\bar\varepsilon - N)^2 - 4\bar\varepsilon,
\end{equation*}
which shows that $N$ increases in the region $N<1+\bar\varepsilon -2\sqrt{\bar\varepsilon}$ and in particular stays positive if $\bar\varepsilon$ is sufficiently small.
There are two options, as long as $|w|\leq 1$:
\begin{itemize}
\item If $N(\bar\rho, y_0) < 1+\bar\varepsilon -2\sqrt{\bar\varepsilon}$, then $N(\rho, y_0) \geq N(\bar\rho, y_0) \geq 1-\bar\varepsilon$ for $\rho \geq \bar\rho$.
\item If $N(\bar\rho, y_0) \geq 1+\bar\varepsilon -2\sqrt{\bar\varepsilon}$, then $N(\rho, y_0) \geq 1+\bar\varepsilon -2\sqrt{\bar\varepsilon}$ for $\rho\geq\bar\rho$.
\end{itemize}
Hence, for any $\varepsilon > 0$, we can choose $\bar\varepsilon$ small enough that
\begin{align*}
N(\rho, y_0) &\geq 1 - \max\{ \bar\varepsilon, 2\sqrt{\bar\varepsilon} - \bar\varepsilon \} \geq 1-\varepsilon\\[0.1cm]
\kappa(\rho, y_0) &\leq 2 + \bar\varepsilon - N(\rho,y_0) \leq 1 + \bar\varepsilon + \max\{ \bar\varepsilon, 2\sqrt{\bar\varepsilon} - \bar\varepsilon \} \leq 1 + \varepsilon
\end{align*}
for all $\rho \geq \bar\rho$ as long as $|w|\leq 1$, and all orbits with initial data $y_0$ in $|y_0-x_0|\leq \delta$, as desired.

Since the $x_0$-orbit has exactly $n$ zeros of $w$ and $|w|\to1$, we can again shrink $\delta$ and increase $\bar\rho$ to ensure that the $y_0$-orbits such that either $y_0 \in \mathscr U_\delta(r_0,x_0)$, or $y_0 \in \mathscr W_\delta(r_0,x_0)$ with $w_0\not=0$,
satisfy $(1-w(\bar\rho,y_0)^2)^2 \leq \varepsilon$ and have exactly $n$ zeros of $w$ for $0 < \rho < \bar\rho$, cf.\ Remark \ref{remark-close-orbits}.
The autonomous energy (\ref{eq-aut-energy}) satisfies
$$\dot F = -4(\kappa-2N) \dot w^2 \geq 4(1-3\varepsilon)\dot w^2.$$
In particular, this shows that $F$ is non-decreasing for $\varepsilon < 1/3$, so
$F(\rho, y_0) \geq F(\bar\rho,y_0) \geq -\varepsilon$, or equivalently
\begin{equation}\label{eq-dw-lower-bound}
2\dot w^2 \geq (1-w^2)^2 - \varepsilon,
\end{equation}
for $\rho\geq\bar\rho$ with $|w|\leq 1$.
Now there are three options for the $y_0$-orbit (we omit $y_0$ as an argument henceforth), where we assume without loss of generality that $w$ and $\dot w$ are negative at $\bar\rho$, in view of the symmetry $(w,U)\mapsto -(w,U)$:
\begin{enumerate}
\item $w$ continues decreasing but stays in the region $|w|\leq1$, hence the orbit is in $\mathscr R_n(r_0)$,
\item $w$ continues decreasing and enters the region $w < -1$, hence the orbit is in $\mathscr E_n(r_0)$,
\item $w$ decreases until it reaches a minimum at some point $\rho_1 \geq \bar \rho$, so $\dot w(\rho_1) = 0$ and $-1 < w(\rho_1) < 0$, and then turns back towards the region $w > 0$. 
\end{enumerate}
To complete the proof of (iii), we thus need to show that $w$ enters the region $w > 1$ in case (iii), so the orbit is in $\mathscr E_{n+1}(r_0)$.
The argument below is, in essence, very similar to \cite[Proposition 4.8]{smol-wass-2}.
Define $b > a \geq \rho_1$ as the first points such that $w(a) = -\sqrt{1-\sqrt\varepsilon}$ and $w(b) = 0$.
These are well-defined because $\dot w$ increases and $\dot w > 0$ in the region $w < 0$ for $\rho > \rho_1$ by (\ref{eq-w}), so $w$ also increases uniformly in that region.
Then for $a \leq \rho \leq b$, the right-hand side of (\ref{eq-dw-lower-bound}) is non-negative, so we may take the square root and estimate
\begin{align*}
F(\rho) &\geq -\varepsilon + 4(1-3\varepsilon)\int_{\rho_1}^\rho \dot w^2 \\
&\geq -\varepsilon +  2\sqrt 2(1-3\varepsilon)\int_{a}^{b} \dot w \sqrt{(1-w^2)^2 - \varepsilon}\\
&= -\varepsilon + 2\sqrt 2(1-3\varepsilon) \int_{-\sqrt{1-\sqrt\varepsilon}}^{0} \sqrt{(1-w^2)^2 - \varepsilon} \,\mathrm dw,
\end{align*}
for $\rho \geq b$.
The right-hand side depends continuously on $\varepsilon$ and tends to $\frac{4\sqrt 2}{3}$ as $\varepsilon \to 0$.
This implies that $F$ is lower bounded by a positive constant for $\rho \geq b$ when $\varepsilon$ is sufficiently small.
Hence, the same is true for $\dot w$, implying thats $w$ keeps increasing uniformly and eventually enters the region $w > 1$, as desired.
\end{proof}

\begin{remark*}
We would like to point out that in \cite[Proposition 3.5]{smol-wass-3}, the authors prove the analogue of part (ii), $r_0 > 1$, of Theorem \ref{thm-ngbh-initial-sets} for their setting.
However, their proof is much more involved, because they treat $r$ as the independent variable, so they cannot use continuous dependence on initial data since the initial point for the $x_0$-orbit is a singular point of the system in those coordinates. However, this singularity only appears as a consequence of using $r$ as a coordinate, and in our system we do not face the same difficulties, so the proof is much simpler. 
\end{remark*}

\section{Energy theorem}
\label{sec-energy-theorem}
% !TEX root = main.tex

Having studied neighbourhoods of different types of orbits, we also wish to understand how the extremal parts of the set $\mathscr I_0$ look like.
This can viewed as a compactness result, as it will provide us with appropriate upper and lower bounds for the shooting method.

\begin{theorem}[Energy theorem]\label{thm-energy-large-smol}
Fix $r_0 > 0$, let $(w_0, U_0) \in \mathscr{I}_0(r_0)$ and
\begin{equation*}
E_0 = E_0(w_0,U_0) = 1+2U_0^2 - \frac{(1-w_0^2)^2}{r_0^2}.
\end{equation*}

\begin{enumerate}
\item[(i)] If $E_0 = 0$, then either $(r_0,w_0,U_0) = (1,0,0) \in \mathscr O$ or $(w_0,U_0) \in \mathscr C(r_0)$ with $r_0 < 1$.
\item[(ii)] If $E_0$ is sufficiently large, then $w$ is monotone and $(w_0,U_0) \in (\mathscr E_0 \cup \mathscr E_1)(r_0)$.
\end{enumerate}
\end{theorem}

\begin{remark*}
In part (ii), the required magnitude of $E_0$ depends on $r_0$.
The heuristic idea behind part (ii) is that if the initial velocity $\dot w_0$ of $w$ is chosen large enough (this is in fact equivalent to chosing the initial energy $E_0$ large enough since $r_0$ is fixed), then $w$ escapes the strip $|w|\leq 1$.
Even though this seeems intuitively obvious, the proof is surprisingly difficult.
The first issue is that, a priori, the orbit could crash arbitrarily fast for large $E_0$, before $w$ escapes the strip $|w|\leq 1$.
The second difficulty is that, even if the orbit does not crash, equation (\ref{eq-U}) implies that the larger $E_0$ is, the faster $U$ decreases, so it could become negative arbitrarily fast, in particular if $\kappa-N$ becomes large.
The crux of the issue here is that we do not have a bound for $\kappa$ that is uniform in $E_0$ (in particular, the bound from Lemma \ref{lemma-ineq} is not good enough).
\end{remark*}

\begin{proof}[Proof of Theorem \ref{thm-energy-large-smol} (i)]
Let $E_0 = 0$ and recall that $\zeta_0 = \sqrt{E_0} =  0$.
By (\ref{eq-zeta}), we then see that $\zeta \equiv 0$.
In view of this, Lemma \ref{lemma-sing} implies that the region $N < 0$ is invariant and the orbit is singular if it enters it.

In fact, we see from (\ref{eq-N-2}) that $\dot N_0 = -2U_0^2$, which is negative if $U_0 \not=0$, so that $N<0$ for $\rho > 0$ in this case.
On the other hand, if $U_0 = 0$, then $\dot N_0 = \ddot N_0 = 0$, but $$\dddot N_0 = -\frac{4w_0^2(1-w_0^2)}{r_0^2},$$
so that either $N < 0$ for $\rho > 0$, or else $(r_0,w_0,U_0)=(1,0,0) \in \mathscr O$, cf.\ \S \ref{subsec-trivial-solutions}.

This shows that the non-trivial solutions with $E_0=0$ are singular, so it only remains to show that they crash rather than escape the strip $|w|\leq 1$.
But this follows immediately since $|w_0|<1$ (note that $E_0 = 0$ prohibits $|w_0|=1$), and $w$ cannot cross the lines $|w|=1$ in the region $N\leq-\zeta \equiv 0$, cf.\ (\ref{eq-w-cannot-cross-1}).
\end{proof}

%%% N LOWER BOUND

As already hinted, the proof of Theorem \ref{thm-energy-large-smol} (ii) is much more involved.
We will need the following lemma, which allows us to control $r$ uniformly for all choices of initial parameters with fixed $r_0$.

\begin{lemma}\label{lemma-r-bound}
Fix $r_0 > 0$.
For each $\varepsilon > 0$, there exists a $\delta > 0$ such that, for any solution of the initial value problem {\upshape (\ref{eq-r}--\ref{eq-zeta}, \ref{eq-initial-vals})} with respect to initial data $x_0 = (w_0,U_0) \in \mathscr I_0(r_0)$, we have
\begin{equation*}
|r(\rho,x_0)-r_0| \leq \varepsilon,
\end{equation*}
for all $0 \leq \rho \leq \delta$ for which $|w|\leq 1$.
\end{lemma}

\begin{remark*}
The main point of the result is that $\delta$ can be chosen independently of $(w_0,U_0)$.
\end{remark*}

\begin{proof}
By Lemma \ref{lemma-ineq}, we know that $N \leq 1$, so we only need a lower bound that is uniform for $(w_0,U_0) \in \mathscr I_0(r_0)$.
To this end, we compare the equations (\ref{eq-r}, \ref{eq-N}) for $(r,N)$ with a simpler system.
For $r_0 < 1$, consider the solution of the initial value problem%
\footnote{
It appears that (\ref{eq-rN-lower-bound-ivp}) does not have an explicit solution, but we would like to note that the solution satisfies $\bar N < 0$ for $\rho>0$ until it stops existing at $\bar r = 0$.
}
\begin{equation}\label{eq-rN-lower-bound-ivp}
\left\{
\begin{array}{ll}
\dot{\bar{r}} = \bar r \bar N,          &\quad \bar r(0) = r_0,\\[0.2cm]
\dot{\bar{N}} = 1 - \frac{1}{\bar r^2}, &\quad \bar N(0) = 0,
\end{array}
\right.
\end{equation}
and define
\begin{equation*}
(\bar r, \bar N) = 
\begin{cases}
(r_0, 0) & \text{if } r_0 \geq 1,\\[0.1cm]
\text{solution of }\mathrm{(\ref{eq-rN-lower-bound-ivp})} & \text{if } r_0 < 1.
\end{cases}
\end{equation*}
We claim that for any $x_0=(w_0,U_0) \in \mathscr I_0(r_0)$, we have
\begin{equation*}
r(\rho, x_0) \geq \bar r(\rho) \quad\text{and}\quad N(\rho, x_0) \geq \bar N(\rho)
\end{equation*}
for all $\rho$ for which $(\bar r, \bar N)$ is well-defined and $|w(\rho,x_0)|\leq 1$.

If $r_0 \geq 1$, then the region $N\geq 0$ is invariant  (cf.\ Lemma \ref{lemma-N-positive-r0-geq1}) and the result follows trivially. 
So assume $r_0 < 1$ and take any $(w_0,U_0) \in \mathscr{I}_0(r_0)$.
Define
\begin{equation*}
\xi = r - \bar r \quad\text{and}\quad \eta = N - \bar N.
\end{equation*}
Then $\xi(0) = \eta(0) = 0$ and using (\ref{eq-r}, \ref{eq-rN-lower-bound-ivp}) we can also calculate
\begin{equation*}
\xi = r_0 \left( \exp\int_0^\rho N - \exp\int_0^\rho \bar N \right) = r_0 \bar r\left(\exp \int_0^\rho \eta - 1\right).
\end{equation*}
This shows that $\xi > 0$ (at least) as long as $\eta > 0$.
We first wish to show that $\eta > 0$ at least for sufficiently small $\rho > 0$.
To this end, we have by (\ref{eq-N}, \ref{eq-rN-lower-bound-ivp})
\begin{equation*}
\dot \eta(0) = \frac{1-(1-w_0^2)^2}{r_0^2},
\end{equation*} 
so that $\dot \eta(0) > 0$ if $|w_0|\not=0$.
If $w_0 = 0$, then $\dot \eta(0) = 0$ gives no information, so we continue differentiating (\ref{eq-N}, \ref{eq-rN-lower-bound-ivp}) to find $\ddot \eta(0) = 0$ but
\begin{equation*}
\dddot \eta(0) = 4U_0^2 - 2(1+2U_0^2)\left(1-\frac{1}{r_0^2}\right) \geq 2\left(\frac{1}{r_0^2}-1\right) > 0,
\end{equation*}
since $r_0 < 1$. Thus $\eta$ increases initially in both cases and thus $\eta > 0$ for sufficiently small $\rho > 0$.

Suppose now that $\eta$ ever reaches zero again, so that there exists a point $\bar \rho > 0$ with $\eta(\bar\rho)=0$ and $\eta > 0$ for $0 < \rho < \bar\rho$.
This implies that $N(\bar\rho)=\bar N(\bar\rho) < 0$, and by (\ref{eq-N}, \ref{eq-rN-lower-bound-ivp}), $\dot \eta$ satisfies
\begin{equation*}
\dot \eta(\bar\rho) 
= \left(\frac{1}{\bar r^2} - \frac{(1-w^2)^2}{r^2} - \kappa N\right)\bigg|_{\bar\rho}
\geq \frac{r + \bar r}{(r\bar r)^2}\bigg|_{\bar\rho} \, \xi(\bar\rho)
= \frac{r_0\bar r(r + \bar r)}{(r\bar r)^2}\bigg|_{\bar\rho} \left(\exp \int_0^{\bar\rho} \eta - 1\right),
\end{equation*}
where we also use the fact that $|w|\leq 1$ and $\kappa \geq 0$ (cf.\ Lemma \ref{lemma-ineq}).
But since $\eta>0$ for $0 < \rho < \bar\rho$, we see that the latter is positive, so such a point cannot exist and $\eta > 0$ for all $\rho > 0$. 
Thus also $\xi > 0$ and in particular $r > \bar r$ and $N > \bar N$ for $\rho > 0$.

We thus have $\bar N \leq N \leq 1$, and since $\bar N$ is a fixed function independent of the initial data $(w_0,U_0) \in \mathscr I_0(r_0)$ (but depending on $r_0$), Lemma \ref{lemma-r-bound} follows immediately.
\end{proof}

\begin{proof}[Proof of Theorem \ref{thm-energy-large-smol} (ii)]
Fix $r_0 > 0$ and use Lemma \ref{lemma-r-bound} to find a $\delta > 0$ so that
\begin{equation}\label{eq-r-bounds}
\frac{1}{2} \leq \frac{r}{r_0} \leq 2,
\end{equation}
as long as $|w|\leq 1$, for any orbit with initial data $(w_0,U_0) \in \mathscr I_0(r_0)$.
For later purposes we also shrink $\delta$ if necessary, to ensure that%
\footnote{Though these assumptions seem arbitrary at the moment, particularly the latter one, their significance will become apparent as we go along.}
\begin{equation}\label{eq-delta-assumption}
\delta < \min\left\{1, \, \frac{r_0}{2(r_0+16)} \log\left(1 + \frac{r_0+16}{2^{13} \, 3 r_0}\right) \right\}.
\end{equation}
Now define
\begin{equation*}
\bar\rho = \bar\rho(\delta,w_0,U_0) = \sup\{ 0 \leq s \leq \delta \,\mid\, |w(\rho)|\leq 1 \text{ for } 0 \leq \rho \leq s \},
\end{equation*}
so that the bounds (\ref{eq-r-bounds}) hold for $0 \leq \rho \leq \bar\rho$.
We will show that $w$ is strictly monotone for $0 \leq \rho\leq \bar\rho$ and that $\bar\rho < \delta$ for sufficiently large $E_0$, which implies the desired result.

In view of the symmetry $(w,U)\mapsto-(w,U)$ and the fact that large $E_0$ necessitates large $U_0^2$, we can without loss of generality assume that $U_0$ is positive (on the other hand, $w_0$ could be negative, but satisfies $|w_0|\leq1$).
We begin by observing that (\ref{eq-w}) can equivalently be written as
\begin{equation*}
\frac{\mathrm d}{\mathrm d\rho}\left(\frac{U}{r\zeta}\right) = -\frac{w(1-w^2)}{r^2\zeta},
\end{equation*}
and an integration gives
\begin{equation}\label{eq-U-formula}
U = \frac{U_0}{\zeta_0}\frac{r}{r_0}\,\zeta - r\zeta \int_0^\rho \frac{w(1-w^2)}{r^2\zeta}.
\end{equation}
Observe that $E_0 \to\infty$ if and only if $U_0 \to \infty$ (since $r_0$ is fixed and $|w_0|\leq 1$), and
\begin{equation*}
\frac{U_0}{\zeta_0} =  \frac{U_0}{\sqrt{E_0}} \to \frac{1}{\sqrt 2} \quad\text{as}\quad E_0\to\infty,
\end{equation*}
so for sufficiently large $E_0$, we have $\frac{1}{2} \leq U_0/\zeta_0 \leq 1$.
Thus, using also the bounds (\ref{eq-r-bounds}), we can estimate the first term in (\ref{eq-U-formula}) as
\begin{equation*}
\frac{1}{4}\,\zeta \leq \frac{U_0}{\zeta_0}\frac{r}{r_0}\,\zeta \leq 2\zeta,
\end{equation*}
for $0 \leq \rho \leq \bar\rho$.
For the second term in (\ref{eq-U-formula}), we note that $r\zeta$ does not increase, since
\begin{equation*}
\frac{\mathrm d}{\mathrm d\rho}\log(r\zeta)= N-\kappa \leq 0
\end{equation*}
by Lemma \ref{lemma-ineq}.
Since $|w|\leq 1$ for $0 \leq \rho\leq \bar\rho$, we get 
\begin{equation*}
\left|r\zeta \int_0^\rho \frac{w(1-w^2)}{r^2\zeta}\right| \leq r\zeta\int_0^\rho \frac{1}{r^2\zeta} \leq \int_0^\rho \frac{1}{r} \leq \frac{2}{r_0},
\end{equation*}
where we use the bounds (\ref{eq-r-bounds}) and also the assumption $\rho\leq \delta < 1$.
It follows that, for $0 \leq \rho \leq \bar\rho$,
\begin{equation*}
\frac{1}{4} \,\zeta - \frac{2}{r_0}  \leq U \leq 2\left(\zeta + \frac{1}{r_0}\right),
\end{equation*}
This then gives (since $\dot w =rU$)
\begin{equation}\label{eq-w-bounds-near-0-large-E0}
\dot w \geq \frac{r_0}{8} \, \zeta - 4
\quad\text{and}\quad 
2U^2 \leq 16\left(\zeta^2 + \frac{1}{r_0}\right),
\end{equation}
where we also use the estimate $(x+y)^2 \leq 2(x^2+y^2)$ to achieve the second inequality and again the bounds (\ref{eq-r-bounds}) for the first one.
If we define $\phi(\rho) = \int_0^\rho \zeta$, then we see by integrating the first inequality in (\ref{eq-w-bounds-near-0-large-E0}) that
\begin{equation}\label{eq-phi-bound}
\phi \leq \frac{8}{r_0} (w-w_0+4\rho) \leq \frac{48}{r_0} \quad \text{for} \quad 0 \leq \rho \leq \bar\rho,
\end{equation}
since $|w|\leq 1$ and $\rho\leq \delta < 1$.
On the other hand, since $\dot\zeta = -\kappa\zeta$ by (\ref{eq-zeta}), we get
\begin{equation*}
\frac{\mathrm d}{\mathrm d\rho}\left(\frac{\kappa}{\zeta}\right) = \frac{1 + 2U^2}{\zeta} \leq \left(1+\frac{16}{r_0}\right)\frac{1}{\zeta} + 16\zeta,
\end{equation*}
where we also use the second inequality from (\ref{eq-w-bounds-near-0-large-E0}).
Since $\kappa(0)=0$, an integration gives
\begin{equation*}
\kappa 
\leq \left(1+\frac{16}{r_0}\right)\zeta \int_0^\rho \frac{1}{\zeta} + 16\zeta\int_0^\rho \zeta
\leq \left(1+\frac{16}{r_0}\right)\rho + 16\zeta\phi
\leq b + a \zeta,  
\end{equation*}
where the second step follows since $\zeta$ decreases, whereas in the last step we use (\ref{eq-phi-bound}) to obtain the positive constants $a = \frac{2^8 3}{r_0},\, b = 1+\frac{16}{r_0}$, both depending only on the fixed $r_0$.
Using $\dot\zeta = -\kappa\zeta$ again, we may then integrate to get
\begin{equation*}
\rho \geq \int_0^\rho \frac{\kappa}{b+a\zeta} = -\int_{\sqrt{E_0}}^{\zeta(\rho)} \frac{\mathrm d\zeta}{\zeta(b+a\zeta)}
= \frac1b\log\left(\frac{b}{\zeta} + a\right) - \frac1b\log\left(\frac{b}{\sqrt{E_0}} + a\right),
\end{equation*}
and solving for $\zeta$ yields the lower bound
\begin{equation}\label{eq-zeta-bound}
\zeta \geq \frac{b}{a} \left[ \left(1 + \frac{b}{a\sqrt{E_0}} \right) e^{b\rho}  - 1\right]^{-1}.
\end{equation}
By (\ref{eq-w-bounds-near-0-large-E0}), $w$ is monotone (at least) as long as
$\frac{r_0}{8} \zeta - 4 > 0$ (and $0 \leq \rho \leq \bar\rho$).
By (\ref{eq-zeta-bound}), this is achieved (at least) for
\begin{equation*}
\rho 
< \frac{1}{b} \log\left( 1 + \frac{\frac{1}{32r_0} - \frac{1}{\sqrt{E_0}}}{\frac{a}{b} + \frac{1}{\sqrt{E_0}}} \right)
\to \frac{1}{b} \log\left( 1 + \frac{b}{32r_0a} \right)
=
\frac{r_0}{r_0+16} \, \log\left(1 + \frac{r_0+16}{2^{13} \, 3  r_0}\right),
\end{equation*}
where the limit is taken as $E_0\to\infty$.
In particular, we see that $w$ is monotone on $0 \leq \rho \leq \bar\rho$ for sufficiently large $E_0$, since $\bar\rho \leq \delta$, and $\delta$ is less than half of the limit above, cf.\ (\ref{eq-delta-assumption}).
Integrating now (\ref{eq-zeta-bound}) for the final time, we get
\begin{equation*}
\phi(\bar\rho) \geq \frac{1}{a} \log\left( 1 + \frac{a}{b}(1-e^{-b\bar\rho}) \sqrt{E_0} \right).
\end{equation*}
But $\phi(\bar\rho)$ is bounded as $E_0 \to \infty$ by (\ref{eq-phi-bound}), so this inequality generates a contradiction unless $\bar\rho \to 0$, which in turn implies that $\bar\rho < \delta$ for sufficiently large $E_0$ and completes the proof.
\end{proof}

\begin{remark*}
The proof actually shows that the orbit escapes the strip $|w|\leq 1$ arbitrarily fast as $E_0\to\infty$, just as one would intuitively expect. However, we do not need this fact so it is not a part of the statement of Theorem \ref{thm-energy-large-smol}. 
\end{remark*}

\section{Construction of symmetric wormholes}
\label{sec-construction}
% !TEX root = main.tex

Having performed the analysis of the previous two sections, we finally have all the ingredients to show the main existence result (Theorem \ref{thm-sym-wh}) of this manuscript.
The theorem will follow by reflection using the symmetries discussed in Remark \ref{remark-symmetries} and the following result.

\begin{theorem}\label{thm-forward-sequence}
Fix $r_0 > 0$.
For each $n \geq 0$, there exist initial parameters
$$0 < w_0^{(n)} = w_0^{(n)}(r_0) \leq 1 \quad\text{and}\quad U_0^{(n)} = U_0^{(n)}(r_0) > 0$$ 
such that a
$$(w_0^{(n)},0) \in \mathscr R_n(r_0)\quad\text{and}\quad(0, U_0^{(n)}) \in \mathscr R_n(r_0).$$
Furthermore, as $n \to \infty$, we have
$$(w_0^{(n)},0) \to (w_0^{(\infty)},0) \in \mathscr O(r_0)\quad\text{and}\quad  (0,U_0^{(n)}) \to (0,U_0^{(\infty)}) \in \mathscr O(r_0).$$
\end{theorem}

The existence of the sequence $w_0^{(n)}$ is an analogue of the corresponding result for black hole initial conditions \cite[Theorem 34]{breit-forg-mais}.
The existence of the sequence $U_0^{(n)}$ has no analogue for the black hole nor particle-like setting.
The final part of the theorem is not really relevant for the construction of wormholes, but it comes at virtually no additional cost.
In any case, it is interesting to know that non-trivial oscillatory orbits exist for $r_0 < 1$, while $w_0^{(n)},U_0^{(n)} \to 0$ for $r_0 \geq 1$.

In Figure \ref{fig-data-sets}, we present a visualization of the set of admissible initial data, where either $U_0=0$ or $w_0=0$ are kept fixed.
The reader can use these figures as a visualization guide throughout the proof.

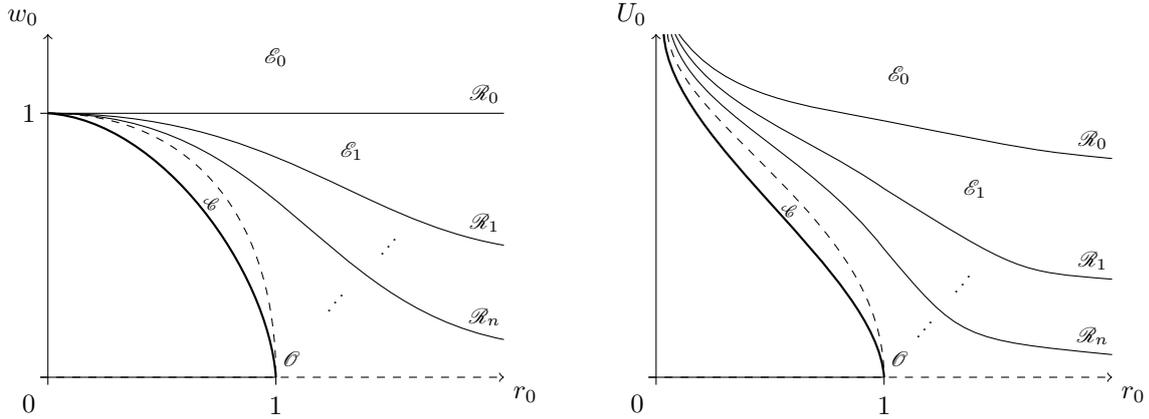
\begin{figure}
\makebox[\linewidth][c]{%
\centering
\begin{subfigure}{0.6\textwidth}
\centering
\begin{tikzpicture}
\draw[->] (0,-0.1) -- (0,4.55) node[above left] {$w_0$};
\draw (-0.1,0) -- (3,0);
\draw[->, dashed] (-0.1,0) -- (6,0) node[anchor = north west] {$r_0$};
\draw (-0.1, 3.5) -- (0.1, 3.5) node[left=4pt] {$1$};
\draw (3, -0.1) -- (3, 0) node[below=4pt] {$1$};
\node at (-0.25,-0.35) {$0$};

\node at (3,4.25) (E0) {\footnotesize $\mathscr E_0$};

\draw (0,3.5) -- (6,3.5);
\node at (5.75,3.75) (R0) {\footnotesize $\mathscr R_0$};

\node at (4,3) (E1) {\footnotesize $\mathscr E_1$};

\draw (0,3.5) .. controls (3,3.5) and (4,2.1) .. (6,1.75);
\node at (5.75,2.05) (R1) {\footnotesize $\mathscr R_1$};

\node[rotate=50] at (4.5, 1.75) {\tiny $\cdots$};
 
\draw (0,3.5) .. controls (3,3.5) and (3.5,1) .. (6,0.5);
\node at (5.75,0.8) (R2) {\footnotesize $\mathscr R_n$};

\node[rotate=50] at (3.8,1) {\tiny $\cdots$};

\draw[dashed] (0,3.5) .. controls (2,3.5) and (3,2) .. (3,0);
\node at (3.2, 0.25) (O) {\footnotesize $\mathscr O$};
%\draw[->] (3.15,0.4) -- (3.05,0.5);
%\draw[->] (3.35, 0.2) -- (3.45,0.1);

\draw[thick] (0,3.5) .. controls (1.5,3.4) and (2.9,1.4) .. (3,0);
\node at (2.15, 2.3) (C) {\tiny $\mathscr C$};

\end{tikzpicture}
\end{subfigure}%
\begin{subfigure}{0.6\textwidth}
\begin{tikzpicture}
\draw[->] (0,-0.1) -- (0,4.55) node[above left] {$U_0$};
\draw (-0.1,0) -- (3,0);
\draw[->, dashed] (-0.1,0) -- (6,0) node[anchor = north west] {$r_0$};
\draw (3, -0.1) -- (3, 0) node[below=4pt] {$1$};
\node at (-0.25,-0.35) {$0$};

\node at (3.2,4) (E0) {\footnotesize $\mathscr E_0$};

\draw (0.3, 4.55) .. controls (0.9, 3.7) and (2, 3.6) .. (3, 3.4) .. controls (5, 3) .. (6,2.9);
\node at (5.75,3.15) (R0) {\footnotesize $\mathscr R_0$};

\node at (4.2,2.5) (E1) {\footnotesize $\mathscr E_1$};

\draw (0.25, 4.55) .. controls (0.7, 3.6) and (2, 3.2) .. (3, 2.5) .. controls (4.75,1.4) .. (6,1.3);
\node at (5.75,1.55) (R0) {\footnotesize $\mathscr R_1$};

\draw (0.2, 4.55) .. controls (0.5, 3.5) and (2, 3) .. (3, 1.7) .. controls (4,0.5) .. (6,0.3);
\node at (5.75,0.55) (R0) {\footnotesize $\mathscr R_n$};

\node[rotate=50] at (4.05, 1.25) {\tiny $\cdots$};

\draw[dashed] (0.15,4.55) .. controls (0.3,3.3) and (3,2.1) .. (3,0);
\node at (3.2, 0.25) (O) {\footnotesize $\mathscr O$};

\node[rotate=50] at (3.55, 0.65) {\tiny $\cdots$};

\draw[thick] (0.1,4.55) .. controls (0.2,3.1) and (2.9,1.4) .. (3,0);
\node at (1.75, 2.2) (C) {\tiny $\mathscr C$};
\end{tikzpicture}
\end{subfigure}
}
\caption{A (qualitative) visualization of the set $\mathscr I_0$ of admissible initial data with $U_0 = 0$ (left) and $w_0=0$ (right). The thick solid lines represent the set where $E_0=0$, the dashed lines represent the set of oscillatory orbits $\mathscr O$, and the thin solid lines represent the regular orbits $\mathscr R_n$.}
\label{fig-data-sets}
\end{figure}

\begin{proof}
We construct the sequences inductively. 
As in the statement, we will assume that for the even orbits (i.e.\ those with $U_0 = 0$), we have $w_0 \geq 0$, and for the odd orbits (i.e.\ those with $w_0=0$), we have $U_0 \geq 0$.
This is possible due to the symmetry $(w,U) \mapsto -(w,U)$.
Note that the boundary condition $E_0 = 0$ of the set $\mathscr I_0(r_0)$ of initial parameters can only be reached if $r_0 \leq 1$. 
We set
\begin{equation*}
w_0^{(\min)} = \begin{cases}
\sqrt{1-r_0}, & \text{if } r_0 < 1,\\[0.1cm]
0, & \text{if } r_0 \geq 1,
\end{cases}
\quad\text{and}\quad
U_0^{(\min)} = \begin{cases}
\sqrt{\frac{1}{2} ( r_0^{-2}-1 )}, & \text{if } r_0 < 1,\\[0.1cm]
0, & \text{if } r_0 \geq 1,
\end{cases}
\end{equation*}
so that $(w_0,0) \in \mathscr I_0(r_0)$ for $w_0^{(\min)} \leq w_0 \leq 1$ and $(0,U_0) \in \mathscr I_0(r_0)$ for $U_0^{(\min)} \leq U_0$.
Furthermore, by Theorem \ref{thm-energy-large-smol} (i),
\begin{equation*}
(w_0^{(\min)},0), (0,U_0^{(\min)}) \in \begin{cases}
\mathscr C(r_0), & \text{if } r_0 < 1,\\[0.1cm]
\mathscr O(r_0), & \text{if } r_0 \geq 1.
\end{cases}
\end{equation*}
The idea is now to perform a shooting method along the lines
\begin{align*}
\mathscr W(r_0) &= \mathscr I(r_0)\cap \{w_0\not=0, U_0=0\},\\[0.1cm]
\mathscr U(r_0) &= \mathscr I(r_0)\cap \{w_0=0, U_0 \not= 0\}.
\end{align*}

Let%
\footnote{With this choice we want to not only show that there are orbits in $\mathscr R_0(r_0)$ but we also want to choose the smallest such, as it is not clear whether they are unique. The same will be done when constructing the orbits in $\mathscr R_n(r_0)$ for $n \geq 1$. In contrast, in Figure \ref{fig-data-sets} we display these sets as if they were unique, because the numerical results suggest that this is indeed the case.}
\begin{align*}
w_0^{(0)} &= \inf\{ w_0 \,\mid\, w_0^{(\min)} \leq w_0, \; (w_0, 0) \in (\mathscr R_0 \cup \mathscr E_0)(r_0) \},\\[0.1cm]
U_0^{(0)} &= \inf\{ U_0 \,\mid\, U_0^{(\min)} \leq U_0, \; (0, U_0) \in (\mathscr R_0 \cup \mathscr E_0)(r_0) \}.
\end{align*}
Note that $w_0^{(0)}$ is well-defined because $(1,0) \in \mathscr R_0(r_0)$ is the trivial solution with $w\equiv 1$, while $U_0^{(0)}$ is well-defined because $(0,U_0) \in \mathscr E_0(r_0)$ for sufficiently large $U_0$ by Theorem \ref{thm-energy-large-smol} (ii).
By Theorem \ref{thm-ngbh-initial-sets}, $(w_0^{(0)}, 0)$ and $(0, U_0^{(0)})$ cannot be in:
\begin{itemize}
\item $\mathscr E_n(r_0)$ for $n\geq 0$, or $\mathscr C(r_0)$, because these sets are open relative to the lines $\mathscr W(r_0)$ and $\mathscr U(r_0)$, cf.\ Theorem \ref{thm-ngbh-initial-sets} (i),
\item $\mathscr O(r_0)$ because this set is neighboured by orbits with arbitrarily many zeros of $w$, cf.\ Theorem \ref{thm-ngbh-initial-sets} (ii),
\item $\mathscr R_n(r_0)$ for $n \geq 1$ because each of these sets respectively is neighboured (relative to the lines $\mathscr W(r_0)$ and $\mathscr U(r_0)$) by orbits with either $n$ or $n+1$ zeros of $w$, cf.\ Theorem \ref{thm-ngbh-initial-sets} (iii).
\end{itemize}
Thus, $(w_0^{(0)},0)$ and $(0, U_0^{(0)})$ must belong to the last remaining option, namely $\mathscr R_0(r_0)$.
We also see in particular that $w_0^{(0)} > w_0^{(\min)}$ and $U_0^{(0)} > U_0^{(\min)}$.
Furthermore, by construction and Theorem \ref{thm-ngbh-initial-sets} (iii), the $(w_0,0)$- and $(0,U_0)$-orbits with sufficiently close $w_0 < w_0^{(0)}$ and $U_0 < U_0^{(0)}$ must be in $\mathscr E_1$.

Now for the induction step, suppose $w_0^{(n-1)}, U_0^{(n-1)}$ have been defined, and let
\begin{align*}
w_0^{(n)} &= \inf\{ w_0 \,\mid\, w_0^{(\min)} \leq w_0 \leq w_0^{(n-1)}, \; (w_0, 0) \in (\mathscr R_{n} \cup \mathscr E_{n})(r_0) \},\\[0.1cm]
U_0^{(n)} &= \inf\{ U_0 \,\mid\, U_0^{(\min)} \leq U_0 \leq U_0^{(n-1)}, \; (0, U_0) \in (\mathscr R_{n} \cup \mathscr E_{n})(r_0)  \}.
\end{align*}
Then $(w_0^{(n)}, 0)$ and $(0,U_0^{(n)})$ belong to $\mathscr R_{n}(r_0)$ by an argument analogous to the base case.
Furthermore, we also have $w_0^{(n)}>w_0^{(\min)}$ and $U_0^{(n)} > U_0^{(\min)}$, and the $(w_0,0)$- and $(0,U_0)$-orbits with sufficiently close $w_0 < w_0^{(n)}$ and $U_0 < U_0^{(n)}$ must be in $\mathscr E_{n+1}$.

Now, $w_0^{(n)}$ and $U_0^{(n)}$ are decreasing bounded sequences, so they converge to some limits $w_0^{(\infty)}$ and $U_0^{(\infty)}$ respectively.
The  $(w_0^{(\infty)},0)$- and $(0,U_0^{(\infty)})$-orbits cannot be in any $\mathscr R_n(r_0)$ or $\mathscr E_n(r_0)$ by construction, nor can they be in $\mathscr C(r_0)$ because this set is open relative to the respective lines.
Hence, they must be in $\mathscr O(r_0)$, completing the proof.
\end{proof}

With this at hand, we can finally show our main result, Theorem \ref{thm-sym-wh}.

\begin{proof}[Proof of Theorem \ref{thm-sym-wh}]
Fix $r_0 > 0$, let $w_0^{(n)} = w_0^{(n)}(r_0) > 0$ and $U_0^{(n)} = U_0^{(n)}(r_0) > 0$ be the sequences manufactured by Theorem \ref{thm-forward-sequence}.

The solutions with initial parameters $(r_0, w_0^{(n)}, 0)$ are then well-defined for all $\rho \geq 0$.
By the symmetries discussed in Remark \ref{remark-symmetries}, we see that the corresponding functions $w,r,\zeta$ are even and $U,N,\kappa$ are odd. 
Therefore each of these solutions is also well-defined backwards for all $\rho < 0$ and shares the same properties (but reflected depending on the parity) of its forward counterpart.
In particular, for the $(w_0^{(n)},0)$-orbit, $w$ in total has $2n$ zeros ($n$ on each half of the real line), and $w \to (-1)^n$ as $\rho \to \pm \infty$.

Similarly, for the solution with initial parameters $(r_0,0,U_0^{(n)})$, we see that $U,r,\zeta$ are even functions while $w,N,\kappa$ are odd, and the solution is also defined for all $\rho < 0$ with the same (but reflected) properties of its forward counterpart.
In particular, for the $(0,U_0^{(n)})$-orbit, $w$ has $2n+1$ zeros ($n$ on each half of the real line, and an additional one at $\rho=0$), and $w \to \pm (-1)^n$ as $\rho \to \pm \infty$.

Now, for each $n \geq 0$, set
\begin{equation*}
x_0^{(n)} = 
\begin{cases}
(r_0, w_0^{(k)}, 0), & \text{if $n = 2k$ is even},\\[0.1cm]
(r_0, 0, U_0^{(k)}), & \text{if $n = 2k+1$ is odd}.
\end{cases} 
\end{equation*}
Recall that we have hidden the metric coefficient $\tau$ by defining $\zeta = \pi_0/(re^\tau)$ where $\pi_0 = r_0e^{\tau_0}\sqrt{E_0}$ is the phantom charge, cf.\ (\ref{eq-breit-variables})
We thus define the metric cofficients
\begin{align*}
\tau^{(n)}(\rho) &=  \log\pi_0 - \log (r\zeta)(\rho,x_0^{(n)})\\[0.1cm]
&= \tau_0 + \log (r_0\sqrt{E_0}) - \log(r\zeta)(\rho,x_0^{(n)}), \\[0.2cm]
r^{(n)}(\rho) &= r(\rho, x_0^{(n)}),
\end{align*}
the Yang-Mills potential coefficient
\begin{align*}
w^{(n)}(\rho) &= w(\rho, x_0^{(n)}),
\end{align*}
and the phantom field (\ref{eq-ph-sol})
\begin{align*}
\phi^{(n)}(\rho) &= \sqrt 2 \int_0^\rho \zeta(s, x_0^{(n)}) \,\mathrm ds.
\end{align*}
Since $\tau_0$ is still free (cf.\ Remark \ref{remark-tau0-free}), we can choose it so as to make $\tau^{(n)} \to 0$ at $\pm\infty$, since $r\zeta$ tends to some finite limit by Theorem \ref{thm-classification}.
Then we see that these functions describe symmetric wormhole solutions, as desired in Theorem \ref{thm-sym-wh}.
Finally, we note that for $r_0 \geq 1$ the variable $N$ has only a single zero located at $\rho=0$ by Lemma \ref{lemma-N-positive-r0-geq1}, and this point describes a non-degenerate wormhole throat.
On the other hand, when $r_0 < 1$ and $n$ is even (resp.\ odd), then $\dot N_0 > 0$ (resp.\ $\dot N_0 < 0$) by (\ref{eq-N}), so the initial point $\rho = 0$ describes a non-degenerate wormhole throat (resp.\ belly).
\end{proof}

\begin{remark}\label{remark-weakness}
Our results have the obvious weakness of not being able to specify the total number of zeros $N$ has for $r_0 < 1$.
One may hope to be able to control the number of zeros of $N$ in the neighbourhood theorem and subsequently the shooting method in a similar way we control the number of zeros of $w$, but this is unfortunately not possible.
The main culprit for this is the fact that nothing prohibits $N$ from having double zeros, i.e.\ points such that $N=0$ and $\dot N = 0$.
In particular, if an orbit has a double zeros of $N$, then for the nearby orbits this zero could bifurcate into several zeros, or even completely disappear.
Even though we cannot prove this using our current techniques, the numerics indicate that the solutions with $r_0<1$ have exactly one non-degenerate throat for even $n$, and exactly two non-degenerate throats for odd $n$.
Nevertheless, Lemma \ref{lemma-N-infinite-zeros} ensures at the very least that $N$ has only a finite zeros for the regular solutions, and in particular the wormhole solutions with $r_0 < 1$ have at most finitely many throats/bellies.
\end{remark}

\section{Asymmetric wormholes}
\label{sec-conclusion-outlook}
% !TEX root = ms.tex

In \cite[\S 4]{hairy-wormholes}, it was proposed that it might be interesting to study also asymmetric solutions, i.e.\ those not possessing the symmetry $\rho\mapsto-\rho$.
In fact, our numerical analysis suggests that there exist such wormhole solutions, having $n$ zeros of $w$ for $\rho > 0$ and $m$ zeros for $\rho \leq 0$.
In Table \ref{tab-asym-pairs}, we list our numerical findings of initial data describing such solutions.
We also plot the corresponding solutions for certain choices of $n$ and $m$ in Figure \ref{fig-asym}.

\begin{table}[ht]
\centering
\pgfplotstableread{asym_pairs.dat}{\asympairs}
\pgfplotstabletypeset[
sort, sort key=0,
display columns/0/.style={precision=0, column name=$n$},
display columns/1/.style={precision=0, column name=$m$},
display columns/2/.style={dec sep align, fixed zerofill, precision=12, column name=$w_0$},
display columns/3/.style={dec sep align, fixed zerofill, precision=12, column name=$U_0$},
every head row/.style={before row=\toprule, after row=\midrule},
every last row/.style={after row=\bottomrule} 
] {\asympairs}
\caption{Initial data describing asymmetric wormholes for $r_0 = 0.75$.}
\label{tab-asym-pairs}
\end{table}

\begin{figure}[t]
\makebox[\linewidth][c]{%
\centering
\begin{subfigure}{0.6\textwidth}
\centering
\begin{tikzpicture}
\begin{axis}[
title=$N^{(n,m)}(\rho)$,
xlabel=$\rho$,
xtick={-25,0,25},
xticklabels={$-\infty$, $0$, $\infty$},
legend cell align= left,
legend pos = north west,
legend style={font=\tiny},
grid=both
]
\addplot [no markers, cyan] table {N_n=0_m=2.dat};
\addlegendentry{$(n,m)=(0,2)$}
\addplot [no markers, magenta] table {N_n=1_m=3.dat};
\addlegendentry{$(n,m)=(1,3)$}
\addplot [no markers, green] table {N_n=2_m=4.dat};
\addlegendentry{$(n,m)=(2,4)$}
\addplot [no markers, gray] table {N_n=3_m=5.dat};
\addlegendentry{$(n,m)=(3,5)$}
\end{axis}
\end{tikzpicture}
\end{subfigure}%
\begin{subfigure}{0.6\textwidth}
\centering
\begin{tikzpicture}
\begin{axis}[
title=$w^{(n,m)}(\rho)$,
xlabel=$\rho$,
y label style={at={(-0.1,1)},rotate=-90,anchor=north},
xtick={-25,0,25},
xticklabels={$-\infty$, $0$, $\infty$},
legend cell align= left,
grid=both
]
\addplot [no markers, cyan] table {w_n=0_m=2.dat};
\addplot [no markers, magenta] table {w_n=1_m=3.dat};
\addplot [no markers, green] table {w_n=2_m=4.dat};
\addplot [no markers, gray] table {w_n=3_m=5.dat};
\end{axis}
\end{tikzpicture}
\end{subfigure}
}
\caption{Asymmetric wormhole solutions for $r_0 = 0.75$.
Plotted using the initial data given in Table \ref{tab-asym-pairs}.
The plots of the other dependent variables are not particularly inspiring so they are omitted.}
\label{fig-asym}
\end{figure}
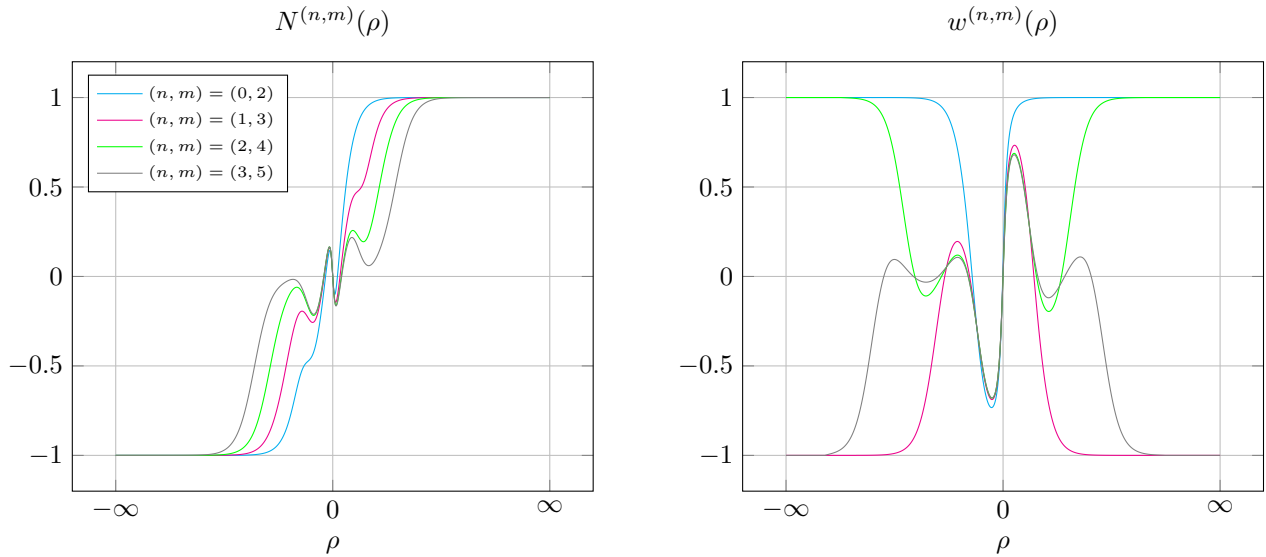

In view of this, we conjecture the following asymmetric extension of Theorem \ref{thm-sym-wh}.

\begin{conjecture}\label{conj-asym-wh}
For each $0 < r_0 < 1$ and integers $0 \leq n < m$, there exists wormhole solution 
\begin{equation*}
(\tau^{(n,m)}, r^{(n,m)}, w^{(n,m)})
\end{equation*}
to the system {\upshape (\ref{eq-r}--\ref{eq-zeta})} such that $w^{(n,m)}$ has $n$ zeros for $\rho > 0$ and $m$ zeros for $\rho \leq 0$.
\end{conjecture}

\begin{remark*}
Note that the even (resp.\ odd) solutions obtained in Theorem \ref{thm-sym-wh} would in this notation have $n=m$ (resp.\ $m=n+1$),
which is why our numerical results only display asymmetric solutions with $m-n\geq 2$.
\end{remark*}

Let us present also some analytical details that could be used to prove Conjecture \ref{conj-asym-wh}.
First, we note that Theorem \ref{thm-forward-sequence} can in fact easily be extended.

\begin{theorem}
Fix $0 < r_0 < 1$.
For each $w_0$ with $r_0+w_0^2 \leq 1$, there exist initial parameters $U_0^{(n)}(w_0) = U_0^{(n)}(r_0,w_0) > 0$ such that
$$ (w_0, U_0^{(n)}(w_0)) \in \mathscr R_n(r_0),\quad\text{and}\quad (w_0,U_0^{(n)}(w_0)) \to (w_0,U_0^{(\infty)}(w_0)) \in \mathscr O(r_0), $$
for each $n \geq 0$ (resp.\ $n\geq 1)$ if $w_0 \geq 0$ (resp.\ $w_0<0$).
\end{theorem}

\begin{proof}
Setting $$U_0^{(\min)}(w_0) = \sqrt{\frac12\left( \frac{(1-w_0^2)^2}{r_0^2} - 1 \right)},$$
we see that $(w_0,U_0^{(\min)}(w_0)) \in \mathscr C(r_0)$ by Theorem \ref{thm-energy-large-smol} (i), since $E_0 = 0$ in this case.
On the other hand, $(w_0,U_0) \in (\mathscr E_0 \cup \mathscr E_1)(r_0)$ and $w$ is monotone for sufficiently large $U_0 > 0$ by Theorem \ref{thm-energy-large-smol} (ii),
so we can apply the same argument as in the proof of Theorem \ref{thm-forward-sequence} to obtain the desired sequence for each fixed $r_0$ and $w_0$, as well as its limit.
\end{proof}

\begin{remark*}
The condition $r_0+w_0^2 \leq 1$ is necessary because we would like to have a $U_0$ for which $(w_0,U_0) \in \mathscr C(r_0)$ in order to perform the shooting method. In the case $r_0 > 1$ of Theorem \ref{thm-forward-sequence}, this was not needed because our set of eligible values of $U_0$ contained the trivial solution $w\equiv 0$ belonging to $\mathscr O(r_0)$, which also serves well enough as a lower bound for the shooting method.
\end{remark*}

We thus obtain a sequence of solutions which are defined for all $\rho \geq 0$ with exactly $n$ zeros of $w$, and which have the correct boundary behaviour at $\rho=\infty$, with neither $w_0$ nor $U_0$ being zero.
However, these might not be defined for all $\rho < 0$ nor do they need to have incorrect behaviour at $\rho=-\infty$ (i.e.\ they might be asymptotically cylindrical rather than flat), since they are no longer symmetric.
Nevertheless, the identities from Remark \ref{remark-symmetries} imply that the backwards solution will have the desired behaviour if
\begin{equation*}
U_0^{(n)}(w_0) = U_0^{(m)}(-w_0), \text{ for some } -\sqrt{1-r_0} \leq w_0 \leq \sqrt{1-r_0} \text{ and } n<m,
\end{equation*}
so the goal is to find $w_0$ having this property.
Due to the way we count zeros of $w$ (cf.\ Definition \ref{def-solution-classes}), it however turns out that the functions $U_0^{(n)}(w_0)$ will have a discontinuity at $w_0=0$ for each $n$, so it makes more sense in this context to define
\begin{equation*}
\tilde U_0^{(n)}(w_0) =
\begin{cases}
U_0^{(n)}(w_0), & \text{if } w_0 \geq 0,\\
U_0^{(n+1)}(w_0), & \text{if } w_0 < 0.
\end{cases}
\end{equation*}
Note that, by construction, $\tilde U_0^{(n)}(w_0)$ decreases with $n$ for fixed $w_0$, since $U_0^{(n)}(w_0)$ does.
The numerics suggest even more.

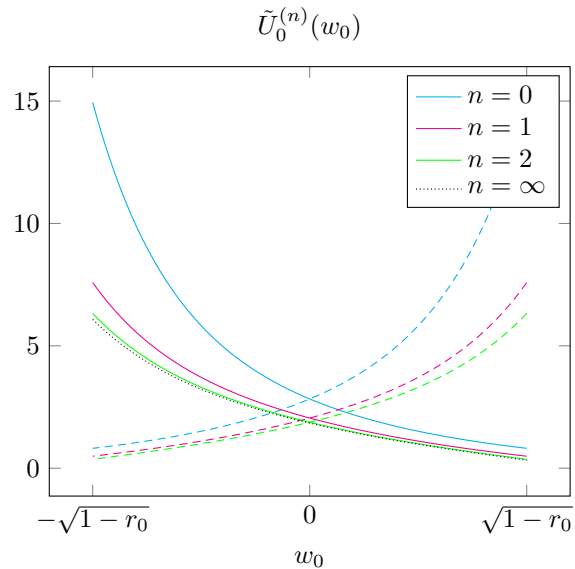
\begin{figure}[t]
\centering
\begin{tikzpicture}
\begin{axis}[
xlabel=$w_0$,
title=$\tilde U_0^{(n)}(w_0)$,
xtick={-0.5,0,0.5},
xticklabels={$-\sqrt{1-r_0}$, $0$, $\sqrt{1-r_0}$},
legend cell align= left
]
\addplot [no markers, cyan] table {mu_0.dat};
\addlegendentry{$n=0$}
\addplot [no markers, cyan, densely dashed, forget plot] table [x expr=-\thisrowno{0}] {mu_0.dat};

\addplot [no markers, magenta] table {mu_1.dat};
\addlegendentry{$n=1$}
\addplot [no markers, magenta, densely dashed, forget plot] table [x expr=-\thisrowno{0}] {mu_1.dat};

\addplot [no markers, green] table {mu_2.dat};
\addlegendentry{$n=2$}
\addplot [no markers, green, densely dashed, forget plot] table [x expr=-\thisrowno{0}] {mu_2.dat};

\addplot [no markers, black, densely dotted] table {mu_inf.dat};
\addlegendentry{$n=\infty$}
\end{axis}
\end{tikzpicture}
\caption{The functions $\tilde U_0^{(n)}(w_0)$ (solid), their reflections $\tilde U_0^{(n)}(-w_0)$ (dashed), and the limit $\tilde U_0^{(\infty)}(w_0)$ (dotted), for $r_0 = 0.75$.  Every intersection between a solid and a dashed curve represents initial data describing a wormhole. The intersections with $w_0=0$ are precisely the odd solutions from Theorem \ref{thm-sym-wh}.}
\label{fig-mu}
\end{figure}

\begin{conjecture}\label{conj-mu}
For each $n \geq 0$, the function $w_0 \mapsto \tilde U_0^{(n)}(w_0)$ is continuous, decreasing, and
\begin{equation*}
\tilde U_0^{(\infty)}(-\sqrt{1-r_0}) \geq \tilde U_0^{(0)}(\sqrt{1-r_0}).
\end{equation*}
\end{conjecture}

If one could show this conjecture, then the desired $w_0$ could be extracted by applying the intermediate value theorem to the functions
\begin{equation*}
\Phi_{m,n} : w_0 \mapsto \tilde U_0^{(n)}(w_0) - \tilde U_0^{(m-1)}(w_0), \quad n < m,
\end{equation*}
because they have opposite signs at each endpoint of the $w_0$-interval.
We visualize this behaviour in Figure \ref{fig-mu}, where each intersection between a solid and a dashed curve represents a zero of some $\Phi_{n,m}$, and hence a wormhole with $n$ zeros of $w$ for $\rho > 0$ and $m$ zeros for $\rho \leq 0$.
The values given in Table \ref{tab-asym-pairs} are precisely these intersections.

Be that as it may, there unfortunately does not seem to be a proof of Conjecture \ref{conj-mu} in sight.
Indeed, we face here similar difficulties as when trying to prove uniqueness of regular orbits with a given number of zeros (with $r_0$ and either $w_0$ or $U_0$ fixed), cf.\ \cite[end of \S 8]{breit-forg-mais} and Remark \ref{remark-uniqueness-regular-orbits}.
We therefore leave the statement merely as a conjecture.

To end the asymmetric discussion, we would also like to recall that we have assumed the initial value $\kappa(0) = 0$ throughout the manuscript.
This is, however, not a necessary condition, and $\kappa(0)$ could be chosen freely.
One should be careful in doing so, because that would require a modification of the proof of the classification given in Theorem \ref{thm-classification}.
Note also that such wormholes would necessarily be asymmetric, so that one should probably first have a good understanding of how these can be constructed for the simpler case $\kappa(0)=0$.

\section{Declarations}
% !TEX root = ms.tex

\subsection{Acknowledgements}
The author would like to express his gratitute to his PhD advisor, Anna Siffert, for suggesting the problem and for the helpful discussions surrounding it.
Many thanks also to Gustav Holzegel for suggesting improvements to the text as well as certain results.

\subsection{Funding}
While working on this project, the author was funded by the Deutsche Forschungsgemeinschaft (DFG, German Research Foundation) under Germany's Excellence Strategy EXC 2044--390685587, Mathematics Münster: Dynamics-Geometry-Structure.
At the time of publishing, the author is funded by the Österreichischer Wissenschaftfonds (FWF) within the project P 36862.
The author is grateful to both sources of funding.

\subsection{Data availability}
The numerical data used for the plots in \S \ref{sec-conclusion-outlook} can be obtained through \href{https://doi.org/10.48550/arXiv.2310.14367}{DOI: 10.48550/arXiv.2310.14367} or by contacting the author.

\subsection{Conflicts of interest}
The author has no conflicts of interest.

\printbibliography

\end{document}